\documentclass[reqno]{amsart}

\usepackage{amsmath, amssymb}
\usepackage[mathscr]{eucal}
\usepackage{upgreek}
\usepackage{hyperref}
\hypersetup{colorlinks=true, linkcolor=blue, urlcolor=blue, citecolor=[rgb]{0, 0.8, 0}}

\theoremstyle{plain}
\newtheorem{theorem}{Theorem}[section]
\newtheorem{lemma}[theorem]{Lemma}

\theoremstyle{remark}
\newtheorem{remark}[theorem]{Remark}

\numberwithin{equation}{section}

\newcommand{\du}{\mathrm{d}}

\newcommand{\iu}{\mathrm{i}}

\DeclareMathOperator{\const}{const}
\DeclareMathOperator{\im}{Im}
\DeclareMathOperator{\ind}{ind}
\DeclareMathOperator{\Res}{Res}

\title[Schr\"{o}dinger operators with distributional potentials]{Schr\"{o}dinger operators with distributional potentials and boundary conditions dependent on the eigenvalue parameter}

\author{Namig J. Guliyev}
\address{Institute of Mathematics and Mechanics, Azerbaijan National Academy of Sciences, 9 B.~Vahabzadeh str., AZ1141, Baku, Azerbaijan.}
\email{njguliyev@gmail.com}



\begin{document}
\maketitle
\begin{abstract}
We study various direct and inverse spectral problems for the one-dimensional Schr\"{o}dinger equation with distributional potential and boundary conditions containing the eigenvalue parameter.
\end{abstract}

\tableofcontents

\section{Introduction} \label{sec:introduction}

At the end of the last millennium, Savchuk and Shkalikov~\cite{SS99} initiated the study of boundary value problems associated with differential equations of the form
\begin{equation} \label{eq:SL}
  - \left( y^{[1]}_s \right)'(x) - s(x) y^{[1]}_s(x) - s^2(x) y(x) = \lambda y(x)
\end{equation}
where $s \in \mathscr{L}_2(0, \pi)$ and $y^{[1]}_s(x) := y'(x) - s(x) y(x)$ denotes the \emph{quasi-derivative} of $y$ with respect to $s$ (the subscript is usually omitted from the notation, but we keep it because in this paper we will consider several potentials simultaneously). This equation formally corresponds to the one-dimensional Schr\"{o}dinger equation with the distributional potential $s' \in \mathscr{W}_2^{-1}(0, \pi)$. Such potentials, especially those describing the so-called point interactions, play an important role in quantum mechanics, solid state physics, atomic and nuclear physics, and electromagnetism \cite{AGHH05}, \cite{AK00}, \cite{KM13}. Direct and inverse spectral problems for boundary value problems generated by the equation~(\ref{eq:SL}) were studied by Savchuk and Shkalikov \cite{S01}, \cite{SS03}, \cite{SS05} and Hryniv and Mykytyuk \cite{HM03}, \cite{HM04a}. More general second-order differential expressions were later considered in \cite{EGNT13a}, \cite{EGNT13b}, \cite{GM10}, \cite{M14}.

In this paper we study various direct and inverse spectral problems for boundary value problems generated by the equation~(\ref{eq:SL}) and the boundary conditions
\begin{equation} \label{eq:boundary}
  \frac{y^{[1]}_s(0)}{y(0)} = -f(\lambda), \qquad \frac{y^{[1]}_s(\pi)}{y(\pi)} = F(\lambda),
\end{equation}
where
\begin{equation} \label{eq:f_F}
  f(\lambda) = h_0 \lambda + h + \sum_{k=1}^d \frac{\delta_k}{h_k - \lambda}, \qquad F(\lambda) = H_0 \lambda + H + \sum_{k=1}^D \frac{\Delta_k}{H_k - \lambda}
\end{equation}
are rational Herglotz--Nevanlinna functions with real coefficients, i.e., $h_0, H_0 \ge 0$, $h, H \in \mathbb{R}$, $\delta_k, \Delta_k > 0$, $h_1 < \ldots < h_d$, $H_1 < \ldots < H_D$. We also include the case when the first (respectively, the second) boundary condition is Dirichlet by writing $f = \infty$ (respectively, $F = \infty$). Similar problems for summable potentials were considered in the author's recent papers \cite{G17}, \cite{G18}, which we closely follow here. However, it is worth emphasizing that, of the results proved in Sections~\ref{sec:direct} and \ref{sec:inverse}, only the oscillation theorem (see Subsection~\ref{ss:oscillation} below) genuinely generalizes the corresponding result of \cite{G17}. The remaining results from \cite{G17}, \cite{G18} depend crucially on the second-order terms (i.e., the terms of order $1/n$) in the asymptotic formulas for the square roots of eigenvalues, and thus do not follow from the results of the current paper.

Eigenvalue problems with boundary conditions dependent on the eigenvalue parameter arise naturally in a variety of physical problems, including heat conduction, diffusion, vibration and electric circuit problems (see \cite{F77}, \cite{F80} and the references therein). Direct and inverse spectral problems of this kind have been studied by many authors (see, e.g., \cite{AOK09}, \cite{BBW02a}, \cite{BBW02b}, \cite{FY10}, \cite{G05}, \cite{IN17}, \cite{K04}, \cite{MC14} for a small selection). In particular, we mention~\cite{ABHM07} in which a spectral problem describing oscillating systems consisting of a continuous part coupled with a discrete part with a finite number of degrees of freedom is studied, and it is shown that this problem is equivalent to a boundary value problem generated by the equation~(\ref{eq:SL}) together with a constant boundary condition at one endpoint and a boundary condition of the form~(\ref{eq:boundary}), (\ref{eq:f_F}) at the other endpoint.

Unlike the case of summable potentials, a little extra care is needed when dealing with inverse problems for distributional potentials. It is easy to see that by adding a constant to $s$ and $f$ and by subtracting the same constant from $F$, we obtain two problems of the form (\ref{eq:SL})-(\ref{eq:boundary}) with the same eigenvalues and eigenfunctions. Therefore some restriction on the coefficients $s$, $f$ and $F$ is necessary. One possible way to tackle this problem is, for example, in the case of constant boundary conditions, to assume that one of (non-Dirichlet) boundary conditions is Neumann, as is done in \cite{HM03}. However for the purposes of this paper, it is more convenient to impose a restriction on the coefficient $s$, by assuming $\int_0^{\pi} s(x) \,\du x = 0$.

The paper is organized as follows. In Section~\ref{sec:preliminaries} we introduce the necessary notation and prove some preliminary lemmas. Section~\ref{sec:transformations} is devoted to transformations between rational Herglotz--Nevanlinna functions and between boundary value problems with distributional potentials having such functions in their boundary conditions. In Subsection~\ref{ss:nevanlinna} we define a transformation between rational Herglotz--Nevanlinna functions and study its properties. In the subsequent three subsections we define direct and inverse transformations between boundary value problems of the form~(\ref{eq:SL})-(\ref{eq:boundary}), study properties of the spectral data under these transformations, and show that these two transformations are, in a sense, inverses of each other. Sections~\ref{sec:direct} and \ref{sec:inverse} are devoted to the solution of various direct and inverse spectral problems. In Subsection~\ref{ss:asymptotics} we obtain asymptotic formulas for the eigenvalues and the norming constants (see Subsection~\ref{ss:hilbert} for the definition) of the problem~(\ref{eq:SL})-(\ref{eq:boundary}). In Subsection~\ref{ss:oscillation} we extend the Sturm oscillation theorem to boundary conditions of the form~(\ref{eq:boundary}). In Subsection~\ref{ss:two_problems} we study further properties of the eigenvalues of a pair of boundary value problems with a common boundary condition and use them in Subsection~\ref{ss:bytwospectra} to solve the two-spectra inverse problem. In Subsection~\ref{ss:byspectraldata} we provide necessary and sufficient conditions for two sequences of real numbers to be the eigenvalues and the norming constants of a problem of the form~(\ref{eq:SL})-(\ref{eq:boundary}). The final Subsection~\ref{ss:byonespectrum} is devoted to inverse problems by one spectrum; namely, we consider symmetric boundary value problems and the Hochstadt--Lieberman theorem for boundary value problems of the form~(\ref{eq:SL})-(\ref{eq:boundary}).

\section{Preliminaries} \label{sec:preliminaries}

\subsection{Notation} \label{ss:notation}

We start by recalling some notation introduced in \cite{G17}. To each function $f$ of the form~(\ref{eq:f_F}) we assign two polynomials $f_\uparrow$ and $f_\downarrow$ by writing this function as
$$
  f(\lambda) = \frac{f_\uparrow(\lambda)}{f_\downarrow(\lambda)},
$$
where
$$
  f_\downarrow(\lambda) := h'_0 \prod_{k=1}^d (h_k - \lambda), \qquad h'_0 := \begin{cases} 1 / h_0, & h_0 > 0, \\ 1, & h_0 = 0. \end{cases}
$$
We define the \emph{index} of $f$ as
$$
  \ind f := \deg f_\uparrow + \deg f_\downarrow.
$$
If $f = \infty$ then we just set
$$
  f_\uparrow(\lambda) := -1, \qquad f_\downarrow(\lambda) := 0, \qquad \ind f := -1.
$$
It is straightforward to check that each nonconstant function $f$ of the form~(\ref{eq:f_F}) is strictly increasing on any interval not containing any of its poles, and $f(\lambda) \to \pm\infty$ (respectively, $f(\lambda) \to h$) as $\lambda \to \pm\infty$ if its index is odd (respectively, even). We denote the smallest pole of $f$ (if it has any) by
$$
  \mathring{\boldsymbol{\uppi}}(f) := \begin{cases} h_1, & \ind f \ge 2, \\ +\infty, & \ind f \le 1, \end{cases}
$$
and the total number of poles of this function not exceeding $\lambda$ by
$$
  \boldsymbol{\Pi}_f(\lambda) := \sum_{\substack{1 \le k \le d \\ h_k \le \lambda}} 1.
$$

For every nonnegative integer $n$ we denote by $\mathscr{R}_n$ the set of rational functions of the form (\ref{eq:f_F}) with $\ind f = n$; we also introduce $\mathscr{R}_{-1} := \{ \infty \}$, which corresponds to the Dirichlet boundary condition. Then $\mathscr{R}_0$ consists of all constant functions, $\mathscr{R}_1$ consists of all increasing affine functions and so on. We also denote
$$
  \mathscr{R} := \bigcup_{n=-1}^\infty \mathscr{R}_n.
$$

We denote by $\mathscr{AC}[0, \pi]$ the set of absolutely continuous functions on $[0, \pi]$. We also denote
\begin{equation*}
  \mathring{\mathscr{L}}_2(0, \pi) := \left\{ g \in \mathscr{L}_2(0, \pi) \biggm| \int_0^{\pi} g(x) \,\du x = 0 \right\}.
\end{equation*}
The notation $x_n = y_n + \ell_2(1)$ means that $\sum_{n = 0}^\infty \left| x_n - y_n \right|^2 < \infty$. Finally, we denote by $\mathscr{P}(s, f, F)$ the boundary value problem (\ref{eq:SL})-(\ref{eq:boundary}), and by $\mathring{\boldsymbol{\uplambda}}(s, f, F)$ the smallest eigenvalue of this problem.

\subsection{Characteristic function} \label{ss:characteristic_function}

Let $\varphi(x, \lambda)$ and $\psi(x, \lambda)$ be the solutions of (\ref{eq:SL}) satisfying the initial conditions
\begin{equation} \label{eq:phi_psi}
  \varphi(0, \lambda) = f_\downarrow(\lambda), \quad \varphi^{[1]}_s(0, \lambda) = -f_\uparrow(\lambda), \quad \psi(\pi, \lambda) = F_\downarrow(\lambda), \quad \psi^{[1]}_s(\pi, \lambda) = F_\uparrow(\lambda),
\end{equation}
and $C(x, \lambda)$ and $S(x, \lambda)$ be the solutions of the same equation satisfying the initial conditions
\begin{equation*}
  C(0, \lambda) = S^{[1]}_s(0, \lambda) = 1, \qquad S(0, \lambda) = C^{[1]}_s(0, \lambda) = 0.
\end{equation*}
Standard arguments show that the eigenvalues of the boundary value problem (\ref{eq:SL})-(\ref{eq:boundary}), which coincide with the zeros of the \emph{characteristic function}
$$\chi(\lambda) := F_\uparrow(\lambda) \varphi(\pi, \lambda) - F_\downarrow(\lambda) \varphi^{[1]}_s(\pi, \lambda) = f_\downarrow(\lambda) \psi^{[1]}_s(0, \lambda) + f_\uparrow(\lambda) \psi(0, \lambda),$$
are real and simple, and for each eigenvalue $\lambda_n$ there exists a unique number $\beta_n \ne 0$ such that
\begin{equation} \label{eq:beta}
  \psi(x, \lambda_n) = \beta_n \varphi(x, \lambda_n).
\end{equation}

Writing $\varphi(x, \lambda)$ as
\begin{equation*}
  \varphi(x, \lambda) = f_\downarrow(\lambda) C(x, \lambda) - f_\uparrow(\lambda) S(x, \lambda),
\end{equation*}
and using the estimates (see, e.g., \cite[Lemma 2.5]{SS03})
\begin{align*}
  C(\pi, \lambda) &= \cos \sqrt{\lambda} \pi + o\left( e^{|\im \sqrt{\lambda}\pi|} \right), & S(\pi, \lambda) &= \frac{\sin \sqrt{\lambda} \pi}{\sqrt{\lambda}} + o\left( \frac{e^{|\im \sqrt{\lambda}\pi|}}{\sqrt{\lambda}} \right), \\
  C^{[1]}_s(\pi, \lambda) &= - \sqrt{\lambda} \sin \sqrt{\lambda} \pi + o\left( \sqrt{\lambda} e^{|\im \sqrt{\lambda}\pi|} \right), & S^{[1]}_s(\pi, \lambda) &= \cos \sqrt{\lambda} \pi + o\left( e^{|\im \sqrt{\lambda}\pi|} \right),
\end{align*}
we calculate
\begin{align}
  \varphi(\pi, \lambda) &= \left( \sqrt{\lambda} \right)^{\ind f} \left( \cos \left( \sqrt{\lambda} + \frac{\ind f}{2} \right) \pi + o\left( e^{|\im \sqrt{\lambda}\pi|} \right) \right), \label{eq:phi_pi} \\
  \varphi^{[1]}_s(\pi, \lambda) &= -\left( \sqrt{\lambda} \right)^{\ind f + 1} \left( \sin \left( \sqrt{\lambda} + \frac{\ind f}{2} \right) \pi + o\left( e^{|\im \sqrt{\lambda}\pi|} \right) \right). \label{eq:phi_1s_pi}
\end{align}
Thus
\begin{equation} \label{eq:chi}
  \chi(\lambda) = \left( \sqrt{\lambda} \right)^{\ind f + \ind F + 1} \left( \sin \left( \sqrt{\lambda} + \frac{\ind f + \ind F}{2} \right) \pi + o\left( e^{|\im \sqrt{\lambda}\pi|} \right) \right).
\end{equation}
Using this estimate, from Hadamard's theorem we obtain (see \cite[Lemma A.1]{G18} for details)
\begin{equation} \label{eq:infinite_product}
  \chi(\lambda) = -\prod_{n < L} (\lambda_n - \lambda) \prod_{n = L} \pi (\lambda_n - \lambda) \prod_{n > L} \frac{\lambda_n - \lambda}{(n - L)^2}
\end{equation}
with
\begin{equation*}
  L := \frac{\ind f + \ind F}{2}.
\end{equation*}

\subsection{Hilbert space formulation and spectral data} \label{ss:hilbert}

In this subsection we will introduce a Hilbert space and construct a self-adjoint operator in it in such a way that the boundary value problem (\ref{eq:SL})-(\ref{eq:boundary}) will be equivalent to the eigenvalue problem for this operator. This construction is related to the theory of self-adjoint exit space extensions of symmetric operators~\cite{G19}. The exact form of the space and the operator will depend on the indices of the functions $f$ and $F$. We will give the details only for odd $\ind f$ and $\ind F$ (i.e., $h_0$, $H_0 > 0$), and then discuss the changes needed in the other cases. When $h_0 > 0$ and $H_0 > 0$ we consider the Hilbert space $\mathcal{H} = \mathscr{L}_2(0,\pi) \oplus \mathbb{C}^{d+D+2}$ with inner product given by

$$\langle Y, Z \rangle := \int_0^{\pi} y(x) \overline{z(x)} \,\du x + \sum_{k=1}^{d} \frac{y_k \overline{z_k}}{\delta_k} + \frac{y_{d+1} \overline{z_{d+1}}}{h_0} + \sum_{k=1}^{D} \frac{\eta_k \overline{\zeta_k}}{\Delta_k} + \frac{\eta_{D+1} \overline{\zeta_{D+1}}}{H_0}$$
for
$$Y = \begin{pmatrix} y(x) \\ y_1 \\ \vdots \\ y_{d+1} \\ \eta_1 \\ \vdots \\ \eta_{D+1} \end{pmatrix}, \qquad Z = \begin{pmatrix} z(x) \\ z_1 \\ \vdots \\ z_{d+1} \\ \zeta_1 \\ \vdots \\ \zeta_{D+1} \end{pmatrix} \in \mathcal{H}.$$
In this space we define the operator
$$A(Y) := \begin{pmatrix} - \left( y^{[1]}_s \right)'(x) - s(x) y^{[1]}_s(x) - s^2(x) y(x) \\ \delta_1 y(0) + h_1 y_1 \\ \vdots \\ \delta_{d} y(0) + h_{d} y_{d} \\ y^{[1]}_s(0) + h y(0) - \sum_{k=1}^{d} y_k \\ H_1 \eta_1 - \Delta_1 y(\pi) \\ \vdots \\ H_{D} \eta_{D} - \Delta_{D} y(\pi) \\ y^{[1]}_s(\pi) - H y(\pi) - \sum_{k=1}^{D} \eta_k \end{pmatrix}$$
with
\begin{multline*}
  D(A) := \left\{ Y \in \mathcal{H} \biggm| y, y^{[1]}_s \in \mathscr{AC}[0,\pi],\ - \left( y^{[1]}_s \right)' - s y^{[1]}_s - s^2 y \in \mathscr{L}_2(0,\pi), \right. \\
  \left. \vphantom{\left( y^{[1]}_s \right)'} y_{d+1} = -h_0 y(0),\ \eta_{D+1} = H_0 y(\pi) \right\}.
\end{multline*}

The necessary modifications for the other cases are as follows. We set $\mathcal{H} = \mathscr{L}_2(0,\pi) \oplus \mathbb{C}^{d+D+1}$ in the case when only one of these numbers equals zero, and $\mathcal{H} = \mathscr{L}_2(0,\pi) \oplus \mathbb{C}^{d+D}$ otherwise. If $h_0 = 0$ (respectively, $H_0 = 0$) we omit the $(d+2)$-th components (respectively, the last components) in the above paragraph, and replace the condition $y_{d+1} = -h_0 y(0)$ (respectively, $\eta_{D+1} = H_0 y(\pi)$) by the condition $y^{[1]}_s(0) + h y(0) - \sum_{k=1}^{d} y_k = 0$ (respectively, $y^{[1]}_s(\pi) - H y(\pi) - \sum_{k=1}^{D} \eta_k = 0$) in the definition of the domain of $A$. If $\ind f \le 0$ (respectively, $\ind F \le 0$), i.e., the first (respectively, the second) boundary condition is independent of the eigenvalue parameter, then there are no $y_k$ (respectively, $\eta_k$) components at all, and the condition $y^{[1]}_s(0) = -h y(0)$ or $y(0) = 0$ (respectively, the condition $y^{[1]}_s(\pi) = H y(\pi)$ or $y(\pi) = 0$) is added in the definition of the domain of $A$.

As in the case of summable potentials, one can prove that the operator $A$ thus defined is self-adjoint, its spectrum is discrete and coincides with the set of eigenvalues of the boundary value problem (\ref{eq:SL})-(\ref{eq:boundary}), and its eigenvectors
$$\Phi_n := \begin{pmatrix} \varphi(x, \lambda_n) \\ \frac{\delta_1}{\lambda_n - h_1} \varphi(0, \lambda_n) \\ \vdots \\ \frac{\delta_{d}}{\lambda_n - h_{d}} \varphi(0, \lambda_n) \\ -h_0 \varphi(0, \lambda_n) \\ \frac{\Delta_1}{H_1 - \lambda_n} \varphi(\pi, \lambda_n) \\ \vdots \\ \frac{\Delta_{D}}{H_{D} - \lambda_n} \varphi(\pi, \lambda_n) \\ H_0 \varphi(\pi, \lambda_n) \end{pmatrix}$$
are orthogonal (see, e.g., \cite{BBW02b}, \cite{F77}). Here, since $\lambda_n = h_m$ if and only if $\varphi(0, \lambda_n) = f_\downarrow(\lambda_n) = 0$, the corresponding component of this vector is well-defined in this case too:
\begin{equation*}
  \frac{\delta_m}{\lambda_n - h_m} \varphi(0, \lambda_n) = -h'_0 \delta_m \prod_{\substack{1 \le k \le d \\ k \ne m}} (h_k - \lambda)
\end{equation*}
(and similarly for $H_m$).

We define the \emph{norming constants} as
\begin{equation*}
\begin{split}
  \gamma_n := \| \Phi_n \|^2 = \int_0^{\pi} \varphi^2(x, \lambda_n) \,\du x &+ f'(\lambda_n) \varphi^2(0, \lambda_n) + F'(\lambda_n) \varphi^2(\pi, \lambda_n) \\
  = \int_0^{\pi} \varphi^2(x, \lambda_n) \,\du x &+ f'_\uparrow(\lambda_n) f_\downarrow(\lambda_n) - f_\uparrow(\lambda_n) f'_\downarrow(\lambda_n) \\
  &+ \frac{1}{\beta_n^2} \left( F'_\uparrow(\lambda_n) F_\downarrow(\lambda_n) - F_\uparrow(\lambda_n) F'_\downarrow(\lambda_n) \right).
\end{split}
\end{equation*}
The numbers $\{ \lambda_n, \gamma_n \}_{n \ge 0}$ are called the \emph{spectral data} of the problem $\mathscr{P}(s, f, F)$. We denote by $\mathring{\boldsymbol{\upgamma}}(s, f, F)$ the first norming constant of the problem $\mathscr{P}(s, f, F)$ (i.e., the norming constant corresponding to the smallest eigenvalue $\mathring{\boldsymbol{\uplambda}}(s, f, F)$ of this problem). As in the regular case, we have the identity (\cite[Lemma~2.1]{G17})
\begin{equation} \label{eq:chi_beta_gamma}
  \chi'(\lambda_n) = \beta_n \gamma_n.
\end{equation}

\subsection{Smallest eigenvalues and nonexistence of zeros} \label{ss:nozeros}

Define a partial order on the set $\mathscr{R}$ as follows: $f \preccurlyeq g$ if and only if either $f = \infty$, or $f$ and $g$ are two functions satisfying $f(\lambda) \le g(\lambda)$ for all $\lambda < \min \{ \mathring{\boldsymbol{\uppi}}(f), \mathring{\boldsymbol{\uppi}}(g) \}$.

\begin{lemma} \label{lem:lambda0}
If $f \preccurlyeq \widetilde{f}$ and $F \preccurlyeq \widetilde{F}$ then $\mathring{\boldsymbol{\uplambda}}(s, f, F) \ge \mathring{\boldsymbol{\uplambda}}(s, \widetilde{f}, \widetilde{F})$. Moreover, for $f$, $F$, $\widetilde{f}$, $\widetilde{F} \in \mathscr{R}_{-1} \cup \mathscr{R}_0$ equality is possible only if $f = \widetilde{f}$ and $F = \widetilde{F}$.
\end{lemma}
\begin{proof}
We will only prove $\mathring{\boldsymbol{\uplambda}}(s, f, F) \ge \mathring{\boldsymbol{\uplambda}}(s, f, \widetilde{F})$; the proof of $\mathring{\boldsymbol{\uplambda}}(s, f, \widetilde{F}) \ge \mathring{\boldsymbol{\uplambda}}(s, \widetilde{f}, \widetilde{F})$ is similar. Let $\nu_0$ be the smallest zero of $\varphi(\pi, \lambda)$. Dividing both sides of the identity
\begin{equation*}
\begin{split}
  \varphi(\pi, \lambda) \varphi^{[1]}_s(\pi, \mu) &- \varphi^{[1]}_s(\pi, \lambda) \varphi(\pi, \mu) \\
  &= f_\uparrow(\lambda) f_\downarrow(\mu) - f_\downarrow(\lambda) f_\uparrow(\mu) + (\lambda - \mu) \int_0^\pi \varphi(t, \lambda) \varphi(t, \mu) \,\du t
\end{split}
\end{equation*}
by $\mu - \lambda$ and taking the limit as $\mu \to \lambda$ we obtain
$$\frac{\du}{\du \lambda} \left( \frac{\varphi^{[1]}_s(\pi, \lambda)}{\varphi(\pi, \lambda)} \right) = - \frac{1}{\varphi^2(\pi, \lambda)} \left( f_\downarrow^2(\lambda) \frac{\du f(\lambda)}{\du \lambda} + \int_0^\pi \varphi^2(t, \lambda) \,\du t \right) < 0$$
for $\lambda \in (-\infty, \nu_0)$. The asymptotics (\ref{eq:phi_pi}) and (\ref{eq:phi_1s_pi}) and the definition of $\nu_0$ imply
$$\lim_{\lambda \to -\infty} \frac{\varphi^{[1]}_s(\pi, \lambda)}{\varphi(\pi, \lambda)} = +\infty, \qquad \lim_{\lambda \to \nu_0-0} \frac{\varphi^{[1]}_s(\pi, \lambda)}{\varphi(\pi, \lambda)} = -\infty.$$
Thus $\varphi^{[1]}_s(\pi, \lambda)/\varphi(\pi, \lambda)$ is strictly monotone decreasing from $+\infty$ to $-\infty$ as $\lambda$ increases from $-\infty$ to $\nu_0$, and the claim of the lemma follows from the fact that $\mathring{\boldsymbol{\uplambda}}(s, f, F)$ and $\mathring{\boldsymbol{\uplambda}}(s, f, \widetilde{F})$ are the smallest values of $\lambda$ for which $\varphi^{[1]}_s(\pi, \lambda)/\varphi(\pi, \lambda) = F(\lambda)$ and $\varphi^{[1]}_s(\pi, \lambda)/\varphi(\pi, \lambda) = \widetilde{F}(\lambda)$ respectively.
\end{proof}

\begin{remark} \label{rem:pi_f}
The above proof also shows that $\mathring{\boldsymbol{\uplambda}}(s, f, F) < \min \{ \mathring{\boldsymbol{\uppi}}(f), \mathring{\boldsymbol{\uppi}}(F) \}$.
\end{remark}

\begin{lemma} \label{lem:no_zero}
If $\lambda \le \mathring{\boldsymbol{\uplambda}}(s, f, \infty)$ (respectively, $\lambda \le \mathring{\boldsymbol{\uplambda}}(s, \infty, F)$) then the function $\varphi(x, \lambda)$ (respectively, $\psi(x, \lambda)$) has no zeros in $(0,\pi)$.
\end{lemma}
\begin{proof}
Let $\nu_0$ be defined as in the proof of the preceding lemma and denote by $S_{\pi}(x, \lambda)$ the solution of (\ref{eq:SL}) satisfying the initial conditions $S_{\pi}(\pi, \lambda) = 0$ and $\left( S_{\pi} \right)^{[1]}_s(\pi, \lambda) = 1$. Since $\varphi(x, \nu_0)$ and $S_{\pi}(x, \nu_0)$ are both eigenfunctions of the problem $\mathscr{P}(s, f, \infty)$, they coincide up to a constant factor. The solution $S_{\pi}(x, \lambda)$ has no zeros in $(0,\pi)$ for values of $\lambda$ not greater than the smallest eigenvalue $\mathring{\boldsymbol{\uplambda}}(s, \infty, \infty)$ of the Dirichlet problem for (\ref{eq:SL}). By Lemma~\ref{lem:lambda0}, $\nu_0 \le \mathring{\boldsymbol{\uplambda}}(s, \infty, \infty)$. Thus $S_{\pi}(x, \nu_0)$ and hence $\varphi(x, \nu_0)$ has no zeros in $(0,\pi)$.

Now suppose to the contrary that $\varphi(x, \lambda)$ has zeros in $(0,\pi)$ for some $\lambda \le \nu_0$, and let $x_0$ be its smallest positive zero. Remark~\ref{rem:pi_f} shows that $\varphi(0, \lambda) = f_\downarrow(\lambda) > 0$ and $\varphi(0, \nu_0) = f_\downarrow(\nu_0) > 0$. Thus $\varphi(x, \lambda) > 0$ and $\varphi(x, \nu_0) > 0$ for $x \in (0, x_0)$. Then $\varphi^{[1]}_s(x_0, \lambda) < 0$ (see, e.g., \cite[Lemma 2]{SB09}), and hence
\begin{equation*}
\begin{split}
  0 &> \varphi(x_0, \nu_0) \varphi^{[1]}_s(x_0, \lambda) - \varphi^{[1]}_s(x_0, \nu_0) \varphi(x_0, \lambda) \\
  &= f_\downarrow(\lambda) f_\downarrow(\nu_0) \left( f(\nu_0) - f(\lambda) \right) + (\nu_0 - \lambda) \int_0^{x_0} \varphi(t, \nu_0) \varphi(t, \lambda) \,\du t > 0.
\end{split}
\end{equation*}
This contradiction proves the lemma for $\varphi$. The proof for $\psi$ is similar.
\end{proof}

\subsection{A characterization of $f_\downarrow$ in terms of spectral data} \label{ss:f_downarrow}

The aim of this subsection is to prove an auxiliary lemma that will be needed in Subsection~\ref{ss:bytwospectra}. This lemma characterizes the polynomial $f_\downarrow$ (up to a multiplicative constant) among all nonzero polynomials.

\begin{lemma} \label{lem:zero_sum}
If $\ind f \ge 2$ (i.e., if $f$ has at least one pole) then $p(\lambda) = f_\downarrow(\lambda)$ is the only nonzero polynomial, up to a multiplicative constant, that satisfies the identities
\begin{equation*}
  \sum_{n=0}^\infty \frac{\lambda_n^k p(\lambda_n)}{\gamma_n} = 0, \qquad k = 0, \ldots, d - 1.
\end{equation*}
\end{lemma}
\begin{proof}
From~(\ref{eq:phi_psi}) and~(\ref{eq:beta}) we have
\begin{equation*}
  f_\downarrow(\lambda_n) = \varphi(0, \lambda_n) = \frac{\psi(0, \lambda_n)}{\beta_n}.
\end{equation*}
Together with~(\ref{eq:chi_beta_gamma}) this implies (for sufficiently large $N$)
\begin{equation*}
  \sum_{n=0}^N \frac{\lambda_n^k f_\downarrow(\lambda_n)}{\gamma_n} = \sum_{n=0}^N \Res_{\lambda = \lambda_n} \frac{\lambda^k \psi(0, \lambda)}{\chi(\lambda)} = \frac{1}{2 \pi \iu} \int_{C_N} \frac{\lambda^k \psi(0, \lambda)}{\chi(\lambda)} \,\du \lambda,
\end{equation*}
where $C_N$ denotes the circle of radius
\begin{equation*}
  \left( N - \frac{\ind f + \ind F - 1}{2} \right)^2
\end{equation*}
centered at the origin. Arguing as in Subsection~\ref{ss:characteristic_function} one obtains
\begin{equation*}
  \psi(0, \lambda) = O \left( \left| \sqrt{\lambda} \right|^{\ind F} e^{|\im \sqrt{\lambda}\pi|} \right).
\end{equation*}
On the other hand, from~(\ref{eq:chi}) we get
\begin{equation*}
  \frac{1}{\chi(\lambda)} = O \left( \left| \sqrt{\lambda} \right|^{-(\ind f + \ind F + 1)} e^{-|\im \sqrt{\lambda}\pi|} \right), \qquad \lambda \in \bigcup_N C_N,
\end{equation*}
and thus
\begin{equation*}
  \frac{\lambda^k \psi(0, \lambda)}{\chi(\lambda)} = O \left( \frac{1}{N^{\ind f - 2 k + 1}} \right), \qquad \lambda \in \bigcup_N C_N
\end{equation*}
with $\ind f - 2 k + 1 \ge 3$. Hence
\begin{equation*}
  \lim_{N \to \infty} \int_{C_N} \frac{\lambda^k \psi(0, \lambda)}{\chi(\lambda)} \,\du \lambda = 0,
\end{equation*}
and thus $f_\downarrow(\lambda)$ does indeed satisfy the identities in the statement of the lemma.

To prove the uniqueness (up to a multiplicative constant) part let
\begin{equation*}
  p(\lambda) = \lambda^d + p_{d-1} \lambda^{d-1} + \ldots + p_1 \lambda + p_0
\end{equation*}
be a monic polynomial satisfying the identities in the statement of the lemma. It is easy to see from the asymptotics of the eigenvalues and the norming constants (see Theorem~\ref{thm:asymptotics}, the proof of which does not use the present lemma) that for each $k = 0$, $\ldots$, $d - 1$ the series
\begin{equation*}
  s_k := \sum_{n=0}^\infty \frac{\lambda_n^k}{\gamma_n}
\end{equation*}
converges absolutely. The identities in the statement of the lemma imply the following identities between the numbers $p_i$ and $s_j$:
\begin{equation} \label{eq:p_s}
  \sum_{i=0}^{d-1} p_i s_{i+k} = -s_{d+k}, \qquad k = 0, 1, \ldots, d - 1.
\end{equation}
We consider them as a system of linear equations (with respect to the numbers $p_i$), the matrix of which is the following Hankel matrix:
\begin{equation*}
  \begin{pmatrix}
    s_0     & s_1    & \ldots & s_{d-1} \\
    s_1     & s_2    & \ldots & s_d \\
    \vdots  & \vdots & \ddots & \vdots \\
    s_{d-1} & s_d    & \ldots & s_{2d-2}
  \end{pmatrix}.
\end{equation*}
The quadratic form corresponding to this matrix is positive definite:
\begin{equation*}
  \sum_{i,j=0}^{d-1} s_{i+j} \xi_i \xi_j = \sum_{i,j=0}^{d-1} \sum_{n=0}^\infty \frac{\lambda_n^{i+j} \xi_i \xi_j}{\gamma_n} = \sum_{n=0}^\infty \sum_{i,j=0}^{d-1} \frac{\lambda_n^{i+j} \xi_i \xi_j}{\gamma_n} = \sum_{n=0}^\infty \frac{1}{\gamma_n} \left( \sum_{i=0}^{d-1} \lambda_n^i \xi_i \right)^2 \ge 0
\end{equation*}
with equality if and only if $\sum_{i=0}^{d-1} \lambda_n^i \xi_i = 0$ for all $n$, i.e. $\xi_0 = \ldots = \xi_{d-1} = 0$. Thus the determinant of the above matrix is strictly positive and hence the system~(\ref{eq:p_s}) has a unique solution.
\end{proof}

\section{Transformations} \label{sec:transformations}

In this section, we introduce Darboux-type transformations between problems of the form~(\ref{eq:SL})-(\ref{eq:boundary}). We will apply these transformations in Sections~\ref{sec:direct} and \ref{sec:inverse} to the solution of several direct and inverse spectral problems for~(\ref{eq:SL})-(\ref{eq:boundary}).

\subsection{Transformation of Herglotz--Nevanlinna functions} \label{ss:nevanlinna}

First we start with transformations between rational Herglotz--Nevanlinna functions. These transformations allow one to shift the index of such a function by one in either direction. Note that in the case of distributional potentials we need slightly more general transformations than those defined in~\cite{G17}.

We denote
$$
  \mathcal{S} := \left\{ (\mu, \tau, \rho, f) \in \mathbb{R} \times \mathbb{R} \times \mathbb{R} \times \mathscr{R} \mid \mu < \mathring{\boldsymbol{\uppi}}(f),\ \tau \ge f(\mu) \text{ if } \ind f \ge 0 \right\},
$$
and define the transformation
$$
  \boldsymbol{\Theta} \colon \mathcal{S} \to \mathscr{R},\ (\mu, \tau, \rho, f) \mapsto \widehat{f}
$$
by
$$
  \widehat{f}(\lambda) := \frac{\mu - \lambda}{f(\lambda) - \tau} + \rho.
$$
In the particular case when $f(\lambda) \equiv \tau$ (respectively, $f = \infty$) this is understood as $\widehat{f} := \infty$ (respectively, $\widehat{f}(\lambda) := \rho$). One sees immediately from this definition that
\begin{equation} \label{eq:ThetaTheta}
  \boldsymbol{\Theta}(\mu, \rho, \tau, \boldsymbol{\Theta}(\mu, \tau, \rho, f)) = f
\end{equation}
and (for $f(\lambda) \not\equiv \tau$)
\begin{equation} \label{eq:f_hat_mu_rho}
  \widehat{f}(\mu) \le \rho
\end{equation}
with equality if and only if $\tau > f(\mu)$. The other main properties of this transformation are summarized in the following lemma.
\begin{lemma} \label{lem:f_hat}
The transformation $\boldsymbol{\Theta}$ is well-defined, i.e., $\widehat{f} := \boldsymbol{\Theta}(\mu, \tau, \rho, f) \in \mathscr{R}$. The poles of $f$ and $\widehat{f}$ interlace if $\ind f \ge 2$ and $\ind \widehat{f} \ge 2$ (i.e., if both $f$ and $\widehat{f}$ have poles); moreover, $\mathring{\boldsymbol{\uppi}}(f) < \mathring{\boldsymbol{\uppi}}(\widehat{f})$ if $\tau = f(\mu)$, and $\mathring{\boldsymbol{\uppi}}(f) > \mathring{\boldsymbol{\uppi}}(\widehat{f})$ if $\tau > f(\mu)$. Also, if $\tau = f(\mu)$ then $\ind \widehat{f} = \ind f - 1$,
\begin{equation} \label{eq:r_hat}
  \widehat{f}_\uparrow(\lambda) = \frac{\rho f_\uparrow(\lambda) - \left( \lambda - \mu + \tau \rho \right) f_\downarrow(\lambda)}{\lambda - \mu}, \qquad \widehat{f}_\downarrow(\lambda) = \frac{f_\uparrow(\lambda) - \tau f_\downarrow(\lambda)}{\lambda - \mu},
\end{equation}
while if $\tau > f(\mu)$ then $\ind \widehat{f} = \ind f + 1$,
\begin{equation*}
  \widehat{f}_\uparrow(\lambda) = -\rho f_\uparrow(\lambda) + \left( \lambda - \mu + \tau \rho \right) f_\downarrow(\lambda), \qquad \widehat{f}_\downarrow(\lambda) = -f_\uparrow(\lambda) + \tau f_\downarrow(\lambda).
\end{equation*}
\end{lemma}
\begin{proof}
The cases $\ind f = -1$, $0$, $1$ are trivial, so we assume that $\ind f \ge 2$. We can write $\widehat{f}$ as
$$
  \widehat{f}(\lambda) = \frac{f_\downarrow(\lambda) (\lambda - \mu)}{\tau f_\downarrow(\lambda) - f_\uparrow(\lambda)} + \rho,
$$
where the polynomials $f_\uparrow$ and $f_\downarrow$, and thus $f_\downarrow$ and $\tau f_\downarrow - f_\uparrow$ have no common roots. When $\tau = f(\mu)$ the polynomial $\tau f_\downarrow(\lambda) - f_\uparrow(\lambda)$ is divisible by $\lambda - \mu$, and hence $\widehat{f}$ is a rational function whose poles $\widehat{h}_1$, $\widehat{h}_2$, $\ldots$, $\widehat{h}_{\widehat{d}}$ coincide with the set $\{ \lambda \ne \mu \mid f(\lambda) = \tau \}$. Recall that $f$ is strictly increasing on each of the intervals $(-\infty, h_1)$, $(h_1, h_2)$, $\dots$, $(h_{d-1}, h_d)$, $(h_d, +\infty)$. Hence $\widehat{h}_k \in (h_k, h_{k+1})$ for $k = 1$, $\ldots$, $d-1$. Therefore $\widehat{d} = d - 1$ or $\widehat{d} = d$, depending on whether the function $f$ takes the value $\tau$ on the interval $(h_d, +\infty)$ or not. If $\ind f = 2d$ then $f(\lambda) \nearrow h < f(\mu) = \tau$ as $\lambda \to +\infty$, and thus $\widehat{d} = d - 1$. Since the degree of the polynomial $\left( \tau f_\downarrow(\lambda) - f_\uparrow(\lambda) \right) / (\lambda - \mu)$ also equals $d - 1$, the function $\widehat{f}$ can be written as
\begin{equation*}
  \widehat{f}(\lambda) = \widehat{h}_0 \lambda + \widehat{h} + \sum_{k=1}^{\widehat{d}} \frac{\widehat{\delta}_k}{\widehat{h}_k - \lambda}.
\end{equation*}
Here $\widehat{h}_0 > 0$ since $\widehat{f}(\lambda) \to +\infty$ as $\lambda \to +\infty$, and $\widehat{\delta}_k > 0$ since $f(\lambda) \nearrow \tau$ as $\lambda \nearrow \widehat{h}_k$. Therefore $\widehat{f} \in \mathscr{R}$ with $\ind \widehat{f} = 2 \widehat{d} + 1 = \ind f - 1$. Finally, the consideration of the leading coefficients of the polynomials $\left( \lambda - \mu + \tau \rho \right) f_\downarrow(\lambda) - \rho f_\uparrow(\lambda)$ and $\tau f_\downarrow - f_\uparrow$ yields the identities~(\ref{eq:r_hat}). If $\ind f = 2d + 1$ then $\widehat{f}$ has one more pole in $(h_d, +\infty)$, so $\widehat{d} = d$. Also, since $f(\lambda) / \lambda \to h_0$ as $\lambda \to +\infty$, we obtain that $\lim_{\lambda \to +\infty} \widehat{f}(\lambda)$ is finite, i.e. $\widehat{h}_0 = 0$, and $\ind \widehat{f} = 2 \widehat{d} = \ind f - 1$.

The case $\tau > f(\mu)$ can be analyzed in a similar way by taking into account the fact that the set of poles of $\widehat{f}$ is now $\{ \lambda \in \mathbb{R} \mid f(\lambda) = \tau \}$.
\end{proof}

\subsection{Direct transformation between problems} \label{ss:isospectral}

We now introduce our first transformation between boundary value problems of the form~(\ref{eq:SL})-(\ref{eq:boundary}), and study its properties. This transformation reduces the index of each non-Dirichlet boundary coefficient by one. Hence, by applying it a sufficient number of times to a boundary value problem of the form (\ref{eq:SL})-(\ref{eq:boundary}), we will eventually arrive at a problem with boundary conditions independent of the eigenvalue parameter.

The domain $\widehat{\mathcal{S}}$ of our transformation consists of all possible boundary value problems of the form (\ref{eq:SL})-(\ref{eq:boundary}), excluding the case when both boundary conditions are Dirichlet:
$$
  \widehat{\mathcal{S}} := \left\{ (s, f, F) \biggm| s \in \mathring{\mathscr{L}}_2(0, \pi),\ f, F \in \mathscr{R},\ \ind f + \ind F \ge -1 \right\}.
$$
We define the transformation
$$
  \widehat{\mathbf{T}} \colon \widehat{\mathcal{S}} \to \mathring{\mathscr{L}}_2(0, \pi) \times \mathscr{R} \times \mathscr{R},\ (s, f, F) \mapsto (\widehat{s}, \widehat{f}, \widehat{F})
$$
by
\begin{equation} \label{eq:s_f_F_hat}
\begin{gathered}
  \widehat{s} := s - \frac{2 v'}{v} + \frac{2}{\pi} \ln \frac{v(\pi)}{v(0)}, \qquad \widehat{f} := \boldsymbol{\Theta} \left( \Lambda, -\frac{v^{[1]}_s(0)}{v(0)}, -\frac{v^{[1]}_s(0)}{v(0)} + \frac{2}{\pi} \ln \frac{v(\pi)}{v(0)}, f \right), \\
  \widehat{F} := \boldsymbol{\Theta} \left( \Lambda, \frac{v^{[1]}_s(\pi)}{v(\pi)}, \frac{v^{[1]}_s(\pi)}{v(\pi)} - \frac{2}{\pi} \ln \frac{v(\pi)}{v(0)}, F \right),
\end{gathered}
\end{equation}
where
\begin{equation} \label{eq:mu}
  \Lambda := \begin{cases} \lambda_0, & f, F \ne \infty, \\ \lambda_0 - 2, & \text{otherwise} \end{cases} \qquad \text{and} \qquad v(x) := \begin{cases} \varphi(x, \Lambda), & f \ne \infty, \\ \psi(x, \Lambda), & f = \infty \end{cases}
\end{equation}
(the motivation for choosing this particular value for $\Lambda$ can be found in~\cite[Remark 3.4]{G17}). That this transformation is well-defined follows from Remark~\ref{rem:pi_f}, Lemmas~\ref{lem:lambda0}, \ref{lem:no_zero}, \ref{lem:f_hat} and the identity $s - 2 v' / v = -s - 2 v^{[1]}_s / v$.

By Lemma~\ref{lem:f_hat}, $\ind \widehat{f} = \ind f - 1$ if $\ind f \ge 0$, and $\ind \widehat{f} = 0$ if $\ind f = -1$. The same is true for $F$ and $\widehat{F}$. Thus we denote
\begin{equation} \label{eq:I}
  I := \ind f - \ind \widehat{f} = \begin{cases} 1, & \ind f \ge 0, \\ -1, & \ind f = -1 \end{cases}
\end{equation}
and
\begin{equation} \label{eq:J}
  J := \frac{\ind f + \ind F}{2} - \frac{\ind \widehat{f} + \ind \widehat{F}}{2} = \begin{cases} 1, & \ind f, \ind F \ge 0, \\ 0, & \text{otherwise.} \end{cases}
\end{equation}

\begin{theorem} \label{thm:transformation}
If $\{ \lambda_n, \gamma_n \}_{n \ge 0}$ is the spectral data of the problem $\mathscr{P}(s, f, F)$ and $(\widehat{s}, \widehat{f}, \widehat{F}) = \widehat{\mathbf{T}} (s, f, F)$ then the spectral data of the transformed problem $\mathscr{P}(\widehat{s}, \widehat{f}, \widehat{F})$ is
$$
  \left\{ \lambda_n, \frac{\gamma_n}{(\lambda_n - \Lambda)^I} \right\}_{n \ge J}.
$$
\end{theorem}
\begin{proof}
It is straightforward to verify that for every $n \ge J$ (i.e., $\lambda_n \ne \Lambda$) the function
$$
  \varphi'(x, \lambda_n) - \frac{v'(x)}{v(x)} \varphi(x, \lambda_n) = \varphi^{[1]}_s(x, \lambda_n) - \frac{v^{[1]}_s(x)}{v(x)} \varphi(x, \lambda_n)
$$
is an eigenfunction of $\mathscr{P}(\widehat{s}, \widehat{f}, \widehat{F})$ corresponding to the eigenvalue $\lambda_n$. Hence the numbers $\lambda_n$ for $n \ge J$ are eigenvalues of this boundary value problem. In order to prove that there are no other eigenvalues, we first observe that if $\widehat{y}$ is an eigenfunction of $\mathscr{P}(\widehat{s}, \widehat{f}, \widehat{F})$ corresponding to an eigenvalue $\lambda \ne \Lambda$ then $\widehat{y}' + \widehat{y} v' / v$ is an eigenfunction of $\mathscr{P}(s, f, F)$ corresponding to the same eigenvalue $\lambda$. Thus no $\lambda \notin \{\Lambda\} \cup \bigcup_{n \ge J} \{\lambda_n\}$ is an eigenvalue of $\mathscr{P}(\widehat{s}, \widehat{f}, \widehat{F})$. It remains to show that $\Lambda$ is not an eigenvalue of $\mathscr{P}(\widehat{s}, \widehat{f}, \widehat{F})$ either. Suppose the contrary. Since the general solution of the equation $-\left( \widehat{y}^{[1]}_{\widehat{s}} \right)' - \widehat{s} \widehat{y}^{[1]}_{\widehat{s}} - \widehat{s}^2 y = \Lambda \widehat{y}$ is of the form
\begin{equation*}
  \widehat{y}(x) := \frac{1}{v(x)} \left( A + B \int_0^x v^2(t) \,\du t \right)
\end{equation*}
for some constants $A$ and $B$, we have
\begin{equation*}
\begin{aligned}
  \frac{\widehat{y}^{[1]}_{\widehat{s}}(0)}{\widehat{y}(0)} &= \frac{v^{[1]}_s(0)}{v(0)} - \frac{2}{\pi} \ln \frac{v(\pi)}{v(0)} + \frac{B v^2(0)}{A}, \\
  \frac{\widehat{y}^{[1]}_{\widehat{s}}(\pi)}{\widehat{y}(\pi)} &= \frac{v^{[1]}_s(\pi)}{v(\pi)} - \frac{2}{\pi} \ln \frac{v(\pi)}{v(0)} + \frac{B v^2(\pi)}{A + B \int_0^\pi v^2(t) \,\du t}.
\end{aligned}
\end{equation*}
Then from~(\ref{eq:f_hat_mu_rho}) we obtain
\begin{equation*}
  \frac{B v^2(0)}{A} \ge 0, \qquad \frac{B v^2(\pi)}{A + B \int_0^\pi v^2(t) \,\du t} \le 0
\end{equation*}
with equality in the first inequality (respectively, in the second inequality) if and only if $f = \infty$ (respectively, $F = \infty$); in the case when one of these denominators is zero the corresponding inequality should be omitted. Since at least one of $f$ and $F$ is not $\infty$, this is a contradiction.

For the part concerning the norming constants, we observe that the eigenfunction
\begin{equation} \label{eq:phi_hat}
  \widehat{\varphi}_n(x) := \begin{cases}
    \dfrac{1}{\Lambda - \lambda_n} \left( \varphi^{[1]}_s(x, \lambda_n) - \dfrac{v^{[1]}_s(x)}{v(x)} \varphi(x, \lambda_n) \right), & \ind f \ge 0, \\
    \varphi^{[1]}_s(x, \lambda_n) - \dfrac{v^{[1]}_s(x)}{v(x)} \varphi(x, \lambda_n), & \ind f = -1
  \end{cases}
\end{equation}
satisfies the initial condition $\widehat{\varphi}_n(0) = \widehat{f}_\downarrow(\lambda_n)$ and thus
\begin{equation*}
\widehat{\gamma}_n := \int_0^{\pi} \widehat{\varphi}_n^2(x) \,\du x + \widehat{f}'(\lambda_n) \widehat{f}_\downarrow^2(\lambda_n) + \widehat{F}'(\lambda_n) \widehat{\varphi}_n^2(\pi) = \begin{cases} \dfrac{\gamma_n}{\lambda_n - \Lambda}, & \ind f \ge 0, \\ \gamma_n (\lambda_n - \Lambda), & \ind f = -1 \end{cases}
\end{equation*}
(see the proof of \cite[Theorem~3.3]{G17} for details).
\end{proof}

\subsection{An expression for $\mathring{\boldsymbol{\upgamma}}(s, f, F)$} \label{ss:gamma0}

When at least one of $f$ and $F$ is $\infty$ the last theorem expresses the spectral data of the problems $\mathscr{P}(s, f, F)$ and $\mathscr{P}(\widehat{s}, \widehat{f}, \widehat{F})$ in terms of each other. But if $f \ne \infty$ and $F \ne \infty$ then the information about the smallest eigenvalue $\lambda_0$ of $\mathscr{P}(s, f, F)$ and the corresponding norming constant $\gamma_0$ is lost. We will see in the next subsection that they can be given almost arbitrarily; ``almost'' here means that $\lambda_0$ should of course be strictly smaller than the smallest eigenvalue of the problem $\mathscr{P}(\widehat{s}, \widehat{f}, \widehat{F})$ and $\gamma_0$ should be positive. In this subsection we will obtain an expression for $\gamma_0$ in terms of the transformed problem $\mathscr{P}(\widehat{s}, \widehat{f}, \widehat{F})$.

Let $\widehat{C}(x, \lambda)$ and $\widehat{S}(x, \lambda)$ be the solutions of the equation
\begin{equation} \label{eq:SL_hat}
  - \left( y^{[1]}_{\widehat{s}} \right)'(x) - \widehat{s}(x) y^{[1]}_{\widehat{s}}(x) - \widehat{s}^2(x) y(x) = \lambda y(x)
\end{equation}
satisfying the initial conditions
\begin{equation*}
  \widehat{C}(0, \lambda) = \widehat{S}^{[1]}_{\widehat{s}}(0, \lambda) = 1, \qquad \widehat{S}(0, \lambda) = \widehat{C}^{[1]}_{\widehat{s}}(0, \lambda) = 0.
\end{equation*}
It is easy to see that the function $1 / \varphi(x, \lambda_0)$ satisfies the equation (\ref{eq:SL_hat}) and the initial conditions
\begin{equation*}
  \frac{1}{\varphi(0, \lambda_0)} = \frac{1}{f_\downarrow(\lambda_0)}, \qquad \left( \frac{1}{\varphi} \right)^{[1]}_{\widehat{s}}(0, \lambda_0) = -\frac{\rho}{f_\downarrow(\lambda_0)},
\end{equation*}
where
\begin{equation*}
  \rho := f(\lambda_0) + \frac{2}{\pi} \ln \frac{\varphi(\pi, \lambda_0)}{f_\downarrow(\lambda_0)}.
\end{equation*}
Thus
\begin{equation} \label{eq:1_phi}
  \frac{1}{\varphi(x, \lambda_0)} = \frac{1}{f_\downarrow(\lambda_0)} \left( \widehat{C}(x, \lambda_0) - \rho \widehat{S}(x, \lambda_0) \right).
\end{equation}
Since $\widehat{S}(x, \lambda_0)$ and $1 / \varphi(x, \lambda_0)$ are both solutions of the equation (\ref{eq:SL_hat}), their Wronskian is constant:
\begin{multline} \label{eq:S_phi}
  \widehat{S}^{[1]}_{\widehat{s}}(x, \lambda_0) \frac{1}{\varphi(x, \lambda_0)} - \widehat{S}(x, \lambda_0) \left( \frac{1}{\varphi} \right)^{[1]}_{\widehat{s}}(x, \lambda_0) \\
  = \widehat{S}^{[1]}_{\widehat{s}}(0, \lambda_0) \frac{1}{\varphi(0, \lambda_0)} - \widehat{S}(0, \lambda_0) \left( \frac{1}{\varphi} \right)^{[1]}_{\widehat{s}}(0, \lambda_0) = \frac{1}{f_\downarrow(\lambda_0)}.
\end{multline}
From here one readily obtains (for a.e. $x \in [0, \pi]$)
\begin{equation*}
  \varphi^2(x, \lambda_0) = f_\downarrow(\lambda_0) \left( \widehat{S}(x, \lambda_0) \varphi(x, \lambda_0) \right)',
\end{equation*}
and hence (note that the function in parentheses is absolutely continuous)
\begin{equation*}
  \int_0^{\pi} \varphi^2(x, \lambda_0) \,\du x = f_\downarrow(\lambda_0) \widehat{S}(\pi, \lambda_0) \varphi(\pi, \lambda_0).
\end{equation*}
If $\ind f \ge 1$ and $\ind F \ge 1$ then we have
\begin{equation*}
  f'(\lambda_0) = \frac{1}{\rho - \widehat{f}(\lambda_0)}
\end{equation*}
and
\begin{equation*}
\begin{split}
  F'(\lambda_0) &= \left( \frac{\varphi^{[1]}_s(\pi, \lambda_0)}{\varphi(\pi, \lambda_0)} - \frac{2}{\pi} \ln \frac{\varphi(\pi, \lambda_0)}{f_\downarrow(\lambda_0)} - \widehat{F}(\lambda_0) \right)^{-1} \\
  &= \left( \frac{\left( 1 / \varphi \right)^{[1]}_{\widehat{s}}(\pi, \lambda_0)}{1 / \varphi(\pi, \lambda_0)} - \widehat{F}(\lambda_0) \right)^{-1}.
\end{split}
\end{equation*}
Therefore
\begin{multline*}
  \int_0^{\pi} \varphi^2(x, \lambda_0) \,\du x + F'(\lambda_0) \varphi^2(\pi, \lambda_0) \\
  = \varphi(\pi, \lambda_0) \left( f_\downarrow(\lambda_0) \widehat{S}(\pi, \lambda_0) + \varphi(\pi, \lambda_0) \left( \frac{\left( 1 / \varphi \right)^{[1]}_{\widehat{s}}(\pi, \lambda_0)}{1 / \varphi(\pi, \lambda_0)} - \widehat{F}(\lambda_0) \right)^{-1} \right).
\end{multline*}
Using (\ref{eq:1_phi}) and (\ref{eq:S_phi}) we can write the expression in parentheses as
\begin{equation*}
\begin{split}
  f_\downarrow(\lambda_0) &\widehat{S}(\pi, \lambda_0) + \varphi(\pi, \lambda_0) \left( \frac{\left( 1 / \varphi \right)^{[1]}_{\widehat{s}}(\pi, \lambda_0)}{1 / \varphi(\pi, \lambda_0)} - \widehat{F}(\lambda_0) \right)^{-1} \\
  &= \left( \frac{\left( 1 / \varphi \right)^{[1]}_{\widehat{s}}(\pi, \lambda_0)}{1 / \varphi(\pi, \lambda_0)} - \widehat{F}(\lambda_0) \right)^{-1} \\
  &\phantom{{}={}} \times \left( f_\downarrow(\lambda_0) \widehat{S}(\pi, \lambda_0) \frac{\left( 1 / \varphi \right)^{[1]}_{\widehat{s}}(\pi, \lambda_0)}{1 / \varphi(\pi, \lambda_0)} - f_\downarrow(\lambda_0) \widehat{F}(\lambda_0) \widehat{S}(\pi, \lambda_0) + \varphi(\pi, \lambda_0) \right) \\
  &= f_\downarrow(\lambda_0) \left( \frac{\left( 1 / \varphi \right)^{[1]}_{\widehat{s}}(\pi, \lambda_0)}{1 / \varphi(\pi, \lambda_0)} - \widehat{F}(\lambda_0) \right)^{-1} \left( \widehat{S}^{[1]}_{\widehat{s}}(\pi, \lambda_0) - \widehat{F}(\lambda_0) \widehat{S}(\pi, \lambda_0) \right) \\
  &= \frac{f_\downarrow^2(\lambda_0)}{\left( \varkappa - \rho \right) \varphi(\pi, \lambda_0)},
\end{split}
\end{equation*}
where
\begin{equation*}
  \varkappa := \frac{\widehat{C}^{[1]}_{\widehat{s}}(\pi, \lambda_0) - \widehat{C}(\pi, \lambda_0) \widehat{F}(\lambda_0)}{\widehat{S}^{[1]}_{\widehat{s}}(\pi, \lambda_0) - \widehat{S}(\pi, \lambda_0) \widehat{F}(\lambda_0)}.
\end{equation*}
Taking into account (\ref{eq:r_hat}), we finally obtain
\begin{equation*}
\begin{split}
  \gamma_0 &= f'(\lambda_0) f_\downarrow^2(\lambda_0) + \int_0^{\pi} \varphi^2(x, \lambda_0) \,\du x + F'(\lambda_0) \varphi^2(\pi, \lambda_0) \\
  &= \frac{f_\downarrow^2(\lambda_0)}{\rho - \widehat{f}(\lambda_0)} - \frac{f_\downarrow^2(\lambda_0)}{\rho - \varkappa} = \widehat{f}_\downarrow^2(\lambda_0) \left( \widehat{f}(\lambda_0) - \varkappa \right) \frac{\rho - \widehat{f}(\lambda_0)}{\rho - \varkappa}
\end{split}
\end{equation*}
The above identity can be written as
\begin{equation} \label{eq:gamma_0}
  \gamma_0 = \left( \widehat{f}_\uparrow(\lambda_0) - \varkappa \widehat{f}_\downarrow(\lambda_0) \right) \frac{\rho \widehat{f}_\downarrow(\lambda_0) - \widehat{f}_\uparrow(\lambda_0)}{\rho - \varkappa},
\end{equation}
and in this form it will also hold for the case $\ind f = 0$. If $F$ is constant then $\widehat{F} = \infty$, and the above expression for $\varkappa$ is understood as
\begin{equation*}
  \varkappa := \frac{\widehat{C}(\pi, \lambda_0)}{\widehat{S}(\pi, \lambda_0)}.
\end{equation*}

\subsection{Inverse transformation between problems} \label{ss:inverseisospectral}

Our aim in this subsection is to invert the action of the transformation $\widehat{\mathbf{T}}$. As we will see shortly, this cannot be done in a unique way, and in order to determine the original problem one needs some more information, e.g., its smallest eigenvalue $\lambda_0$ and the corresponding norming constant $\gamma_0$. Theorem~\ref{thm:transformation} shows that the smallest eigenvalue is not removed if and only if one of the boundary conditions of the original problem is Dirichlet. In this case the corresponding norming constant is multiplied (respectively, divided) by two if and only if the first (respectively, the second) boundary condition is Dirichlet.

With these considerations in mind, we consider the union
\begin{equation*}
  \widetilde{\mathcal{S}} := \widetilde{\mathcal{S}}_1 \cup \widetilde{\mathcal{S}}_2 \cup \widetilde{\mathcal{S}}_3
\end{equation*}
of the three disjoint sets
\begin{equation*}
  \widetilde{\mathcal{S}}_1 := \left\{ (\mu, \nu, s, f, F) \colon s \in \mathring{\mathscr{L}}_2(0, \pi),\ f, F \in \mathscr{R},\ \mu < \mathring{\boldsymbol{\uplambda}}(s, f, F),\ \nu > 0 \right\},
\end{equation*}
\begin{multline*}
  \widetilde{\mathcal{S}}_2 := \Big\{ (\mu, \nu, s, f, F) \colon s \in \mathring{\mathscr{L}}_2(0, \pi),\ f \in \mathscr{R}_0,\ F \in \mathscr{R}, \\
  \mu = \mathring{\boldsymbol{\uplambda}}(s, f, F),\ \nu = \mathring{\boldsymbol{\upgamma}}(s, f, F) / 2 \Big\},
\end{multline*}
and
\begin{multline*}
  \widetilde{\mathcal{S}}_3 := \Big\{ (\mu, \nu, s, f, F) \colon s \in \mathring{\mathscr{L}}_2(0, \pi),\ f \in \mathscr{R},\ F \in \mathscr{R}_0, \\
  \mu = \mathring{\boldsymbol{\uplambda}}(s, f, F),\ \nu = 2 \mathring{\boldsymbol{\upgamma}}(s, f, F) \Big\}.
\end{multline*}
We define the transformation
$$\widetilde{\mathbf{T}} \colon \widetilde{\mathcal{S}} \to \mathring{\mathscr{L}}_2(0, \pi) \times \mathscr{R} \times \mathscr{R},\ (\mu, \nu, s, f, F) \mapsto (\widetilde{s}, \widetilde{f}, \widetilde{F})$$
by
\begin{equation} \label{eq:s_f_F_tilde}
\begin{gathered}
  \widetilde{s} := s - \frac{2 u'}{u} + \frac{2}{\pi} \ln \frac{u(\pi)}{u(0)}, \qquad \widetilde{f} := \boldsymbol{\Theta} \left( \Lambda, -\frac{u^{[1]}_s(0)}{u(0)}, -\frac{u^{[1]}_s(0)}{u(0)} + \frac{2}{\pi} \ln \frac{u(\pi)}{u(0)}, f \right), \\
  \widetilde{F} := \boldsymbol{\Theta} \left( \Lambda, \frac{u^{[1]}_s(\pi)}{u(\pi)}, \frac{u^{[1]}_s(\pi)}{u(\pi)} - \frac{2}{\pi} \ln \frac{u(\pi)}{u(0)}, F \right),
\end{gathered}
\end{equation}
where
\begin{equation*}
  \Lambda := \begin{cases} \mu, & (\mu, \nu, s, f, F) \in \widetilde{\mathcal{S}}_1, \\ \mu - 2, & (\mu, \nu, s, f, F) \in \widetilde{\mathcal{S}}_2 \cup \widetilde{\mathcal{S}}_3 \end{cases}
\end{equation*}
and
\begin{equation*}
  u(x) := \begin{cases} C(x, \mu) - \rho S(x, \mu), & (\mu, \nu, s, f, F) \in \widetilde{\mathcal{S}}_1, \\ \varphi(x, \mu - 2), & (\mu, \nu, s, f, F) \in \widetilde{\mathcal{S}}_2, \\ \psi(x, \mu - 2), & (\mu, \nu, s, f, F) \in \widetilde{\mathcal{S}}_3 \end{cases}
\end{equation*}
with
\begin{equation*}
  \rho := \frac{\nu \varkappa + f_\uparrow(\mu) \left( \varkappa f_\downarrow(\mu) - f_\uparrow(\mu) \right)}{\nu + f_\downarrow(\mu) \left( \varkappa f_\downarrow(\mu) - f_\uparrow(\mu) \right)}, \qquad \varkappa := \frac{C^{[1]}_s(\pi, \mu) - C(\pi, \mu) F(\mu)}{S^{[1]}_s(\pi, \mu) - S(\pi, \mu) F(\mu)}.
\end{equation*}

The well-definedness of this transformation on $\widetilde{\mathcal{S}}_1$ can be justified as follows; the two other cases are similar and even simpler. Lemma~\ref{lem:lambda0} implies $\mu < \mathring{\boldsymbol{\uplambda}}(s, f, F) \le \mathring{\boldsymbol{\uplambda}}(s, \infty, F)$, i.e., $\mu$ is not an eigenvalue of the problem $\mathscr{P}(s, \infty, F)$. Since the denominator of the above expression for $\varkappa$ is zero only at the eigenvalues of the problem $\mathscr{P}(s, \infty, F)$, $\varkappa$ is well-defined. Arguing as in the proof of Lemma~\ref{lem:lambda0} we see that $\mathscr{P}(s, \varkappa, F)$ has only one eigenvalue not exceeding $\mathring{\boldsymbol{\uplambda}}(s, \infty, F)$, and hence $\mu = \mathring{\boldsymbol{\uplambda}}(s, \varkappa, F)$. The same proof also shows that if $f \ne \infty$ then
$$
  f(\mu) < f \left( \mathring{\boldsymbol{\uplambda}}(s, f, F) \right) = -\frac{\psi^{[1]}_s \left( 0, \mathring{\boldsymbol{\uplambda}}(s, f, F) \right)}{\psi \left( 0, \mathring{\boldsymbol{\uplambda}}(s, f, F) \right)} < -\frac{\psi^{[1]}_s(0, \mu)}{\psi(0, \mu)} = \varkappa.
$$
Thus the denominator of $\rho$ is strictly positive and $f(\mu) < \rho < \varkappa$. Lemma~\ref{lem:lambda0} implies that $\mu = \mathring{\boldsymbol{\uplambda}}(s, \varkappa, F) < \mathring{\boldsymbol{\uplambda}}(s, \rho, \infty)$, and hence Lemma~\ref{lem:no_zero} shows that $u$ has no zeros on $[0,\pi]$. Moreover, if $F \ne \infty$ then using the asymptotics of the solutions $S$ and $S^{[1]}_s$ we obtain that the denominator of the expression for $\varkappa$ is strictly positive, and thus $\rho < \varkappa$ implies $F(\mu) < u^{[1]}_s(\pi) / u(\pi)$. This shows that the arguments of $\boldsymbol{\Theta}$ in the expressions for both $\widetilde{f}$ and $\widetilde{F}$ in~(\ref{eq:s_f_F_tilde}) belong to its domain.

Now we prove that, in a sense, the two transformations that we defined in this and the previous subsections are inverses of each other.

\begin{theorem} \label{thm:inverse}
The transformations $\widehat{\mathbf{T}}$ and $\widetilde{\mathbf{T}}$ are inverses of each other in the sense that if $(s, f, F) \in \widehat{\mathcal{S}}$ and $(\widehat{s}, \widehat{f}, \widehat{F}) = \widehat{\mathbf{T}}(s, f, F)$ then
$$
  \widetilde{\mathbf{T}} \left( \mathring{\boldsymbol{\uplambda}}(s, f, F), \mathring{\boldsymbol{\upgamma}}(s, f, F), \widehat{s}, \widehat{f}, \widehat{F} \right) = (s, f, F),$$
and conversely if $(\mu, \nu, s, f, F) \in \widetilde{\mathcal{S}}$ then $\widehat{\mathbf{T}} \widetilde{\mathbf{T}}(\mu, \nu, s, f, F) = (s, f, F)$.
\end{theorem}
\begin{proof}
The main idea of the proof is to show that the solutions $v$ and $u$, used in (\ref{eq:s_f_F_hat}) and (\ref{eq:s_f_F_tilde}) respectively, are inverses of each other up to a constant factor. We will give the details for the first statement when $f$, $F \ne \infty$; the remaining cases and the second statement can be analyzed in an analogous manner.

Denote $\lambda_0 := \mathring{\boldsymbol{\uplambda}}(s, f, F)$ and $(\widetilde{s}, \widetilde{f}, \widetilde{F}) := \widetilde{\mathbf{T}} \left( \lambda_0, \mathring{\boldsymbol{\upgamma}}(s, f, F), \widehat{s}, \widehat{f}, \widehat{F} \right)$. In our case, $\left( \lambda_0, \mathring{\boldsymbol{\upgamma}}(s, f, F), \widehat{s}, \widehat{f}, \widehat{F} \right) \in \widetilde{\mathcal{S}}_1$ and $v(x) = \varphi(x, \lambda_0)$. Comparing the definition of $\widetilde{\mathbf{T}}$ with the expression (\ref{eq:gamma_0}) for $\mathring{\boldsymbol{\upgamma}}(s, f, F)$, we conclude that both $f_\downarrow(\lambda_0) / v(x)$ and $u(x)$ satisfy
the equation (\ref{eq:SL_hat}) with $\lambda = \lambda_0$ and the same initial conditions, and hence $u(x) = f_\downarrow(\lambda_0) / v(x)$. Thus
$$
  \widetilde{s} = \widehat{s} - \frac{2 u'}{u} + \frac{2}{\pi} \ln \frac{u(\pi)}{u(0)} = s - \frac{2 v'}{v} + \frac{2}{\pi} \ln \frac{v(\pi)}{v(0)} - \frac{2 u'}{u} + \frac{2}{\pi} \ln \frac{u(\pi)}{u(0)} = s.
$$
Finally, the identity~(\ref{eq:ThetaTheta}) implies $\widetilde{f} = f$ and $\widetilde{F} = F$.
\end{proof}

We can also prove an analogue of Theorem~\ref{thm:transformation} for the transformation $\widetilde{\mathbf{T}}$.
\begin{theorem} \label{thm:inverse_transformation}
If $\{ \lambda_n, \gamma_n \}_{n \ge 0}$ is the spectral data of the problem $\mathscr{P}(s, f, F)$ and $(\widetilde{s}, \widetilde{f}, \widetilde{F}) = \widetilde{\mathbf{T}}(\mu, \nu, s, f, F)$ then the spectral data of the problem $\mathscr{P}(\widetilde{s}, \widetilde{f}, \widetilde{F})$ is
$$
  \left\{ \lambda_n, \gamma_n (\lambda_n - \Lambda)^I \right\}_{n \ge -J},
$$
where
\begin{equation*}
  I := \ind \widetilde{f} - \ind f, \qquad J := \frac{\ind \widetilde{f} + \ind \widetilde{F}}{2} - \frac{\ind f + \ind F}{2},
\end{equation*}
and we denote $\lambda_{-1} := \mu$ and $\gamma_{-1} := \nu$ in the case when $J = 1$.
\end{theorem}
\begin{proof}
By Theorem \ref{thm:inverse}, $\widehat{\mathbf{T}}(\widetilde{s}, \widetilde{f}, \widetilde{F}) = (s, f, F)$. Thus the part of the claim concerning the eigenvalues and the norming constants with nonnegative indices immediately follows from Theorem \ref{thm:transformation}. It only remains to consider the case $\mu < \mathring{\boldsymbol{\uplambda}}(s, f, F)$, i.e., $J = 1$. It is straightforward to verify in this case that $\mu$ is an eigenvalue of the problem $\mathscr{P}(\widetilde{s}, \widetilde{f}, \widetilde{F})$ corresponding to the eigenfunction $1 / u$. Again, comparing the definition of $\widetilde{\mathbf{T}}$ with (\ref{eq:gamma_0}) we obtain that $\nu$ is the corresponding norming constant.
\end{proof}

\section{Direct spectral problems} \label{sec:direct}

\subsection{Asymptotics of eigenvalues and norming constants} \label{ss:asymptotics}

In the case of constant boundary conditions the eigenvalues of the problem $\mathscr{P}(s, f, F)$ have the asymptotics (\cite[Theorem 1]{S01}, \cite[Lemma 7.1]{HM03})
\begin{equation*}
  \sqrt{\lambda_n} = n + \frac{\mathcal{N}_{\text{D}}}{2} + \ell_2(1),
\end{equation*}
where $\mathcal{N}_{\text{D}}$ is the number of Dirichlet boundary conditions. The norming constants have the asymptotics
\begin{equation*}
  \gamma_n = \frac{\pi}{2} \left( n + \frac{\mathcal{N}_{\text{D}}}{2} \right)^{-2} \left( 1 + \ell_2(1) \right)
\end{equation*}
if the first boundary condition is Dirichlet, and
\begin{equation*}
  \gamma_n = \frac{\pi}{2} \left( 1 + \ell_2(1) \right)
\end{equation*}
otherwise (see~\cite[Lemmas 2.4 and 7.2]{HM03} but note that the norming constants are defined differently there). Our next theorem shows that the transformation $\widehat{\mathbf{T}}$ allows us to extend these results to the case of boundary conditions dependent on the eigenvalue parameter and write them in a unified manner.

\begin{theorem} \label{thm:asymptotics}
The spectral data of the problem $\mathscr{P}(s, f, F)$ have the asymptotics
$$\begin{aligned}
  \sqrt{\lambda_n} &= n - \frac{\ind f + \ind F}{2} + \ell_2(1), \\
  \gamma_n &= \frac{\pi}{2} \left( n - \frac{\ind f + \ind F}{2} \right)^{2 \ind f} \left( 1 + \ell_2(1) \right).
\end{aligned}$$
\end{theorem}
\begin{proof}
According to the discussion at the beginning of the subsection, these formulas hold in the case of constant boundary conditions. Consider now the chain of problems $\mathscr{P}(s^{(k)}, f^{(k)}, F^{(k)})$ defined by
\begin{equation} \label{eq:P_k}
\begin{aligned}
  (s^{(0)}, f^{(0)}, F^{(0)}) &:= (s, f, F), \\
  (s^{(k)}, f^{(k)}, F^{(k)}) &:= \widehat{\mathbf{T}} (s^{(k-1)}, f^{(k-1)}, F^{(k-1)}), \qquad k = 1, 2, \ldots, K,
\end{aligned}
\end{equation}
where $K := \max \{ \ind f, \ind F \}$. Then the last problem $\mathscr{P}(s^{(K)}, f^{(K)}, F^{(K)})$ has constant boundary conditions, and hence its eigenvalues have the asymptotics
\begin{equation*}
  \sqrt{\lambda_n^{(K)}} = n - \frac{\ind f^{(K)} + \ind F^{(K)}}{2} + \ell_2(1).
\end{equation*}
Let $I$ and $J$ be defined by (\ref{eq:I})-(\ref{eq:J}) with $f$ and $F$ replaced by $f^{(K-1)}$ and $F^{(K-1)}$ respectively. Using Theorem~\ref{thm:transformation} we calculate
\begin{equation*}
\begin{split}
  \sqrt{\lambda_n^{(K-1)}} = \sqrt{\lambda_{n-J}^{(K)}} &= n - J - \frac{\ind f^{(K)} + \ind F^{(K)}}{2} + \ell_2(1) \\
  &= n - \frac{\ind f^{(K-1)} + \ind F^{(K-1)}}{2} + \ell_2(1).
\end{split}
\end{equation*}
Repeating this argument $K-1$ more times we get the above asymptotics for $\sqrt{\lambda_n}$.

In a similar manner, from
$$
  \gamma_n^{(K)} = \frac{\pi}{2} \left( n - \frac{\ind f^{(K)} + \ind F^{(K)}}{2} \right)^{2 \ind f^{(K)}} \left( 1 + \ell_2(1) \right),
$$
Theorem~\ref{thm:transformation} and the asymptotics of the eigenvalues we obtain
\begin{equation*}
\begin{split}
  \gamma_n^{(K-1)} &= \gamma_{n-J}^{(K)} \left( \lambda_n^{(K-1)} - \mu \right)^I \\
  &= \frac{\pi}{2} \left( n - J - \frac{\ind f^{(K)} + \ind F^{(K)}}{2} \right)^{2 \ind f^{(K)}} \\
  &\phantom{{}= \frac{\pi}{2}{}} \times \left( n - \frac{\ind f^{(K-1)} + \ind F^{(K-1)}}{2} \right)^{2I} \left( 1 + \ell_2(1) \right) \\
  &= \frac{\pi}{2} \left( n - \frac{\ind f^{(K-1)} + \ind F^{(K-1)}}{2} \right)^{2 \ind f^{(K-1)}} \left( 1 + \ell_2(1) \right).
\end{split}
\end{equation*}
Again, repeating this argument $K-1$ more times yields the above asymptotics for the sequence $\gamma_n$.
\end{proof}

\subsection{Oscillation of eigenfunctions} \label{ss:oscillation}

As is shown in~\cite[Theorem 1]{SB09} (see also~\cite[Theorem 4.4]{HH14}), the Sturm oscillation theorem holds also in the case of distributional potentials with constant boundary conditions, i.e., an eigenfunction corresponding to the $n$-th eigenvalue of the problem $\mathscr{P}(s, f, F)$ with $\ind f$, $\ind F \le 0$ has exactly $n$ zeros in the open interval $(0, \pi)$. By using the transformation $\widehat{\mathbf{T}}$, we will now extend this result to the case of arbitrary $\ind f$ and $\ind F$. But first we need the following auxiliary result.

\begin{lemma} \label{lem:oscillation}
Let $J$ and $\widehat{\varphi}_n$ be defined by the formulas (\ref{eq:J}) and (\ref{eq:phi_hat}) respectively. If the function $\widehat{\varphi}_n(x)$ has $N$ zeros in $(0, \pi)$ then the function $\varphi(x, \lambda_n)$ has exactly $N + J + \boldsymbol{\Pi}_{\widehat{f}}(\lambda_n) + \boldsymbol{\Pi}_{\widehat{F}}(\lambda_n) - \boldsymbol{\Pi}_f(\lambda_n) - \boldsymbol{\Pi}_F(\lambda_n)$ zeros in $(0, \pi)$.
\end{lemma}
\begin{proof}
We give the proof for the case $f \ne \infty$; in the case when $f = \infty$ we only need to consider $\psi$ instead of $\varphi$. Let $\Lambda$ be defined by~(\ref{eq:mu}). Denote by $x_1$, $\ldots$, $x_N$ the zeros of the function $\widehat{\varphi}_n(x)$ in $(0, \pi)$. Then for each $k = 1$, $\ldots$, $N-1$ we have
$$
  \int_{x_k}^{x_{k+1}} \varphi(x, \lambda_n) \varphi(x, \Lambda) \,\du x = \widehat{\varphi}_n(x_{k+1}) \varphi(x_{k+1}, \Lambda) - \widehat{\varphi}_n(x_k) \varphi(x_k, \Lambda) = 0
$$
and hence the function $\varphi(x, \lambda_n)$ has a zero in $(x_k, x_{k+1})$. Similarly we obtain that between any two zeros of $\varphi(x, \lambda_n)$ there is a zero of $\widehat{\varphi}_n(x)$. This implies that the zeros of these two functions strictly interlace; in particular, they have no common zeros. Thus $\varphi(x, \lambda_n)$ has $N-1$ zeros in $(x_1, x_N)$.

Consider now the interval $(0, x_1)$. It is obvious that if $f_\downarrow(\lambda_n) = \varphi(0, \lambda_n) = 0$ then $\varphi(x, \lambda_n)$ does not have a zero in $(0, x_1)$, and if $\widehat{f}_\downarrow(\lambda_n) = \widehat{\varphi}_n(0) = 0$ then $\varphi(x, \lambda_n)$ has a zero in $(0, x_1)$. In the case when $\widehat{f}_\downarrow(\lambda_n) f_\downarrow(\lambda_n) \ne 0$ we see that if $\varphi(x, \lambda_n)$ does not have a zero in $(0, x_1)$ then
$$
  \int_{0}^{x_1} \varphi(x, \lambda_n) \varphi(x, \Lambda) \,\du x = -\widehat{f}_\downarrow(\lambda_n) f_\downarrow(\Lambda)
$$
and $\varphi(0, \lambda_n) \varphi(0, \Lambda) = f_\downarrow(\lambda_n) f_\downarrow(\Lambda)$ must have the same sign, and if $\varphi(x, \lambda_n)$ has a zero $x_0 \in (0, x_1)$ then
$$
  \int_{0}^{x_0} \frac{\widehat{\varphi}_n(x)}{\varphi(x, \Lambda)} \,\du x = -\frac{f_\downarrow(\lambda_n)}{(\Lambda - \lambda_n) f_\downarrow(\Lambda)}
$$
and $\widehat{\varphi}_n(0) / \varphi(0, \Lambda) = \widehat{f}_\downarrow(\lambda_n) / f_\downarrow(\Lambda)$ must have the same sign (recall that $\Lambda < \lambda_n$). Therefore the function $\varphi(x, \lambda_n)$ has a zero in $(0, x_1)$ if and only if $\widehat{f}_\downarrow(\lambda_n) f_\downarrow(\lambda_n) > 0$ or $\widehat{f}_\downarrow(\lambda_n) = 0$, i.e., if and only if the functions $f$ and $\widehat{f}$ have the same number of poles not exceeding $\lambda_n$. Thus the number of zeros of $\varphi(x, \lambda_n)$ in $(0, x_1)$ equals $1 - \left( \boldsymbol{\Pi}_f(\lambda_n) - \boldsymbol{\Pi}_{\widehat{f}}(\lambda_n) \right)$.

A similar assertion holds for the interval $(x_N, \pi)$ and the functions $F$ and $\widehat{F}$ if the boundary condition at $\pi$ is not Dirichlet (i.e., $J = 1$). Otherwise, if the boundary condition at $\pi$ is Dirichlet (i.e., $J = 0$), the function $\varphi(x, \lambda_n)$ does not have a zero in $(x_N, \pi)$, but $\boldsymbol{\Pi}_F(\lambda_n) = \boldsymbol{\Pi}_{\widehat{F}}(\lambda_n) = 0$. Thus the number of zeros of $\varphi(x, \lambda_n)$ in $(x_N, \pi)$ equals $J - \left( \boldsymbol{\Pi}_F(\lambda_n) - \boldsymbol{\Pi}_{\widehat{F}}(\lambda_n) \right)$. This concludes the proof.
\end{proof}

We can now proceed to our main oscillation result.

\begin{theorem} \label{thm:oscillation}
An eigenfunction of the problem $\mathscr{P}(s, f, F)$ corresponding to the eigenvalue $\lambda_n$ has exactly $n - \boldsymbol{\Pi}_f(\lambda_n) - \boldsymbol{\Pi}_F(\lambda_n)$ zeros in $(0, \pi)$.
\end{theorem}
\begin{proof}
Consider again the problems $\mathscr{P}(s^{(k)}, f^{(k)}, F^{(k)})$ defined by (\ref{eq:P_k}). Since the last problem $\mathscr{P}(s^{(K)}, f^{(K)}, F^{(K)})$ has constant boundary conditions, its eigenfunction corresponding to the eigenvalue $\lambda_m^{(K)}$ has $m$ zeros in the open interval $(0, \pi)$ for each $m \ge 0$. On the other hand, the constancy of $f^{(K)}$ and $F^{(K)}$ implies $\boldsymbol{\Pi}_{f^{(K)}}(\lambda) \equiv 0$ and $\boldsymbol{\Pi}_{F^{(K)}}(\lambda) \equiv 0$, and hence the statement of the theorem holds in this case. Let $J^{(k)}$ be defined by (\ref{eq:J}) with $f$ and $F$ replaced by $f^{(k)}$ and $F^{(k)}$ respectively. By successive applications of Theorem~\ref{thm:transformation}, it follows that $\lambda_n = \lambda_{n - J'}^{(K)}$, where $J' := \sum_{k=0}^{K-1} J^{(k)}$.
Applying Lemma~\ref{lem:oscillation} successively to the problems $\mathscr{P}(s^{(K-1)}, f^{(K-1)}, F^{(K-1)})$, $\ldots$, $\mathscr{P}(s^{(0)}, f^{(0)}, F^{(0)})$, we finally obtain that an eigenfunction of $\mathscr{P}(s, f, F)$ corresponding to the eigenvalue $\lambda_n$ has
\begin{multline*}
  n - J' + \sum_{k=0}^{K-1} \left( J^{(k)} + \boldsymbol{\Pi}_{f^{(k+1)}}(\lambda_n) + \boldsymbol{\Pi}_{F^{(k+1)}}(\lambda_n) - \boldsymbol{\Pi}_{f^{(k)}}(\lambda_n) - \boldsymbol{\Pi}_{F^{(k)}}(\lambda_n) \right) \\ = n - \boldsymbol{\Pi}_f(\lambda_n) - \boldsymbol{\Pi}_F(\lambda_n)
\end{multline*}
zeros in $(0, \pi)$.
\end{proof}

\subsection{On problems with a common boundary condition} \label{ss:two_problems}

Let $\alpha \ne 0$ be some real number and together with the problem $\mathscr{P}(s, f, F)$ consider the problem $\mathscr{P}(s, f + \alpha, F)$. Denote the eigenvalues of the latter problem by $\mu_n$. Theorem~\ref{thm:asymptotics} shows that they have the same asymptotics as $\lambda_n$. In this subsection we study further properties of these two sequences. We will use these results in Subsection~\ref{ss:bytwospectra}.

Throughout this subsection we assume that $\ind f \ge 0$ so that the problems $\mathscr{P}(s, f, F)$ and $\mathscr{P}(s, f + \alpha, F)$ are different. We also assume that no eigenvalue of $\mathscr{P}(s, f, F)$ is a pole of $f$ or, which is the same, the spectra of the problems $\mathscr{P}(s, f, F)$ and $\mathscr{P}(s, f + \alpha, F)$ do not intersect. 
Obviously, $\mu_n$ are the zeros of the (characteristic) function
\begin{equation*}
  \xi(\lambda) := F_\uparrow(\lambda) \theta(\pi, \lambda) - F_\downarrow(\lambda) \theta^{[1]}_s(\pi, \lambda) = f_\downarrow(\lambda) \psi^{[1]}_s(0, \lambda) + \left( f_\uparrow(\lambda) + \alpha f_\downarrow(\lambda) \right) \psi(0, \lambda),
\end{equation*}
where $\theta(x, \lambda)$ is the solution of (\ref{eq:SL}) satisfying the initial conditions
\begin{equation} \label{eq:theta}
  \theta(0, \lambda) = f_\downarrow(\lambda), \qquad \theta^{[1]}_s(0, \lambda) = -f_\uparrow(\lambda) - \alpha f_\downarrow(\lambda).
\end{equation}
Together with the expression for $\chi$ from Subsection~\ref{ss:characteristic_function} this gives
\begin{equation} \label{eq:xi-chi}
  \xi(\lambda) - \chi(\lambda) = \alpha f_\downarrow(\lambda) \psi(0, \lambda).
\end{equation}

The function
\begin{equation*}
  m(\lambda) := -\frac{\xi(\lambda)}{\chi(\lambda)}
\end{equation*}
satisfies the identity $m(\overline{\lambda}) = \overline{m(\lambda)}$ and is a meromorphic function with poles at $\lambda_n$ and zeros at $\mu_n$. For nonreal values of $\lambda$ the solution
\begin{equation*}
  y(x, \lambda) := \theta(x, \lambda) + m(\lambda) \varphi(x, \lambda)
\end{equation*}
satisfies the boundary condition
\begin{equation*}
  F_\uparrow(\lambda) y(\pi, \lambda) - F_\downarrow(\lambda) y^{[1]}_s(\pi, \lambda) = 0.
\end{equation*}
Taking into account~(\ref{eq:phi_psi}) and (\ref{eq:theta}), we obtain
\begin{multline*}
  \left. (\lambda - \mu) \int_0^{\pi} y(x, \lambda) y(x, \mu) \,\du x = \left( y(x, \lambda) y^{[1]}_s(x, \mu) - y^{[1]}_s(x, \lambda) y(x, \mu) \right) \right|_0^{\pi} \\
  = \left( F(\mu) - F(\lambda) \right) y(\pi, \lambda) y(\pi, \mu) + \alpha f_\downarrow(\lambda) f_\downarrow(\mu) \left( m(\lambda) - m(\mu) \right) \\
  + \left( f_\downarrow(\lambda) f_\uparrow(\mu) - f_\downarrow(\mu) f_\uparrow(\lambda) \right) \left( 1 + m(\lambda) \right) \left( 1 + m(\mu) \right).
\end{multline*}
For $\mu = \overline{\lambda}$ this implies
\begin{equation*}
\begin{split}
  \alpha \frac{\im m(\lambda)}{\im \lambda} &= \frac{1}{\left| f_\downarrow(\lambda) \right|^2} \int_0^{\pi} \left| y(x, \lambda) \right|^2 \,\du x \\
  &+ \left| \frac{y(\pi, \lambda)}{f_\downarrow(\lambda)} \right|^2 \frac{\im F(\lambda)}{\im \lambda} + \left| 1 + m(\lambda) \right|^2 \frac{\im f(\lambda)}{\im \lambda} > 0.
\end{split}
\end{equation*}
Thus $\alpha m(\lambda)$ is a Herglotz--Nevanlinna function, and hence its zeros $\mu_n$ and poles $\lambda_n$ interlace.

Using~(\ref{eq:phi_psi}), (\ref{eq:beta}) and the constancy of the Wronskian we obtain
\begin{align*}
  \xi(\lambda_n) &= F_\uparrow(\lambda_n) \theta(\pi, \lambda_n) - F_\downarrow(\lambda_n) \theta^{[1]}_s(\pi, \lambda_n) \\
  &= \beta_n \left( \varphi^{[1]}_s(\pi, \lambda_n) \theta(\pi, \lambda_n) - \varphi(\pi, \lambda_n) \theta^{[1]}_s(\pi, \lambda_n) \right) \\
  &= \beta_n \left( \varphi^{[1]}_s(0, \lambda_n) \theta(0, \lambda_n) - \varphi(0, \lambda_n) \theta^{[1]}_s(0, \lambda_n) \right) = \alpha \beta_n f_\downarrow^2(\lambda_n).
\end{align*}
Together with~(\ref{eq:chi_beta_gamma}) this yields
\begin{equation*}
  \gamma_n = \frac{\alpha f_\downarrow^2(\lambda_n) \chi'(\lambda_n)}{\xi(\lambda_n)}.
\end{equation*}
We will need this formula in order to solve the two-spectra inverse problem in Subsection~\ref{ss:bytwospectra}, but for now we will use it to obtain more refined asymptotics for the difference $\sqrt{\lambda_n} - \sqrt{\vphantom{\lambda_n} \mu_n}$. The mean value theorem implies
\begin{equation} \label{eq:xi_lambda_mu}
  \xi(\lambda_n) = \xi(\lambda_n) - \xi(\mu_n) = \left( \sqrt{\lambda_n} - \sqrt{\vphantom{\lambda_n} \mu_n} \right) \left( \sqrt{\lambda_n} + \sqrt{\vphantom{\lambda_n} \mu_n} \right) \xi'(\zeta_n)
\end{equation}
for $\zeta_n \in \left[ \lambda_n, \mu_n \right]$. Thus
\begin{equation*}
  \sqrt{\lambda_n} - \sqrt{\vphantom{\lambda_n} \mu_n} = \frac{\alpha f_\downarrow^2(\lambda_n) \chi'(\lambda_n)}{\left( \sqrt{\lambda_n} + \sqrt{\vphantom{\lambda_n} \mu_n} \right) \gamma_n \xi'(\zeta_n)}.
\end{equation*}
Using the infinite product representations
\begin{equation*}
  \chi(\lambda) = -\prod_{n < L} (\lambda_n - \lambda) \prod_{n = L} \pi (\lambda_n - \lambda) \prod_{n > L} \frac{\lambda_n - \lambda}{(n - L)^2}
\end{equation*}
and
\begin{equation*}
  \xi(\lambda) = -\prod_{n < L} (\mu_n - \lambda) \prod_{n = L} \pi (\mu_n - \lambda) \prod_{n > L} \frac{\mu_n - \lambda}{(n - L)^2}
\end{equation*}
we obtain (see \cite[Appendix A]{G18} and \cite[Lemma 3.2]{HM04a} for details)
\begin{equation} \label{eq:chi_xi}
  \chi'(\lambda_n) = (-1)^n \left( n - L \right)^{2L} \left( \frac{\pi}{2} + \ell_2(1) \right), \quad \xi'(\zeta_n) = (-1)^n \left( n - L \right)^{2L} \left( \frac{\pi}{2} + \ell_2(1) \right),
\end{equation}
where
\begin{equation*}
  L := \frac{\ind f + \ind F}{2}.
\end{equation*}
Now using the asymptotics of $\gamma_n$ from Theorem~\ref{thm:asymptotics}, we finally obtain
\begin{equation*}
  \sqrt{\lambda_n} - \sqrt{\vphantom{\lambda_n} \mu_n} = \left( n - L \right)^{-2 r - 1} \left( \frac{\alpha \left( h'_0 \right)^2}{\pi} + \ell_2(1) \right),
\end{equation*}
where
\begin{equation*}
  r := \ind f - 2 d = \begin{cases} 1, & \ind f \text{ is odd,} \\ 0, & \ind f \text{ is even.} \end{cases}
\end{equation*}

\section{Inverse spectral problems} \label{sec:inverse}

\subsection{Inverse problem by eigenvalues and norming constants} \label{ss:byspectraldata}

Theorem~\ref{thm:asymptotics} shows that the spectral data of a problem of the form (\ref{eq:SL})-(\ref{eq:boundary}) necessarily satisfies the conditions
\begin{equation} \label{eq:increasing}
  \lambda_0 < \lambda_1 < \lambda_2 < \ldots, \qquad \gamma_n > 0, \quad n \ge 0
\end{equation}
and
\begin{equation} \label{eq:asymptotics}
  \sqrt{\lambda_n} = n - \frac{M + N}{2} + \ell_2(1), \qquad \gamma_n = \frac{\pi}{2} \left( n - \frac{M + N}{2} \right)^{2M} \left( 1 + \ell_2(1) \right)
\end{equation}
for some integers $M$, $N \ge -1$. The aim of this subsection is to prove that these necessary conditions are also sufficient for sequences of real numbers $\{ \lambda_n \}_{n \ge 0}$ and $\{ \gamma_n \}_{n \ge 0}$ to be the eigenvalues and the norming constants of a problem of the form (\ref{eq:SL})-(\ref{eq:boundary}).

If $\{ \lambda_n \}_{n \ge 0}$ and $\{ \gamma_n \}_{n \ge 0}$ are two sequences of real numbers satisfying the above conditions with $-1 \le M$, $N \le 0$, then there exists a unique boundary value problem $\mathscr{P}(s, f, F)$ with constant boundary conditions having these sequences as its spectral data (see, e.g., \cite[Corollary 5.4 and Theorem 7.4]{HM03}). The transformations defined in Section~\ref{sec:transformations} allow us to extend this result to the case of boundary conditions (\ref{eq:boundary}).

\begin{theorem} \label{thm:by_spectral_data}
Let $\{ \lambda_n \}_{n \ge 0}$ and $\{ \gamma_n \}_{n \ge 0}$ be sequences of real numbers satisfying the conditions (\ref{eq:increasing}) and (\ref{eq:asymptotics}). Then there exists a unique boundary value problem $\mathscr{P}(s, f, F)$ having the spectral data $\{ \lambda_n, \gamma_n \}_{n \ge 0}$.
\end{theorem}
\begin{proof}
With Theorem~\ref{thm:transformation} in mind, we denote $K := \max \{ M, N \}$, and consider the numbers $M^{(k)}$, $N^{(k)}$ and the sequences $\{ \lambda_n^{(k)} \}_{n \ge 0}$, $\{ \gamma_n^{(k)} \}_{n \ge 0}$ for $k = 0$, $1$, $\ldots$, $K$ defined by
$$
  M^{(0)} := M, \qquad N^{(0)} := N, \qquad \lambda_n^{(0)} := \lambda_n, \qquad \gamma_n^{(0)} := \gamma_n
$$
and
$$\begin{aligned}
  M^{(k)} &:= M^{(k-1)} - I, \qquad & N^{(k)} &:= N^{(k-1)} + I - 2 J, \\
  \lambda_n^{(k)} &:= \lambda_{n+J}^{(k-1)}, \qquad & \gamma_n^{(k)} &:= \frac{\gamma_{n+J}^{(k-1)}}{(\lambda_{n+J}^{(k-1)} - \lambda_0^{(k-1)} + 2 - 2 J)^I},
\end{aligned}$$
where
$$
  I := \begin{cases} 1, & M^{(k-1)} \ge 0, \\ -1, & M^{(k-1)} = -1, \end{cases} \qquad J := \begin{cases} 1, & M^{(k-1)}, N^{(k-1)} \ge 0, \\ 0, & \text{otherwise} \end{cases}
$$
(we omit the indices of $I$ and $J$ here to avoid double indices). One easily verifies that they satisfy the conditions (\ref{eq:increasing}) and (\ref{eq:asymptotics}) with $M$, $N$, $\lambda_n$ and $\gamma_n$ replaced by $M^{(k)}$, $N^{(k)}$, $\lambda_n^{(k)}$ and $\gamma_n^{(k)}$ respectively. Moreover, one of the numbers $M^{(K)}$ and $N^{(K)}$ is always $0$, while the other one is either $0$ or $-1$. Hence there exists a boundary value problem $\mathscr{P}(s^{(K)}, f^{(K)}, F^{(K)})$ (with constant boundary conditions) having $\{ \lambda_n^{(K)}, \gamma_n^{(K)} \}_{n \ge 0}$ as its spectral data. Now we successively define $\mathscr{P}(s^{(K-1)}, f^{(K-1)}, F^{(K-1)})$, $\dots$, $\mathscr{P}(s^{(0)}, f^{(0)}, F^{(0)})$ by
$$
  (s^{(k-1)}, f^{(k-1)}, F^{(k-1)}) := \widetilde{\mathbf{T}}(\lambda_0^{(k-1)}, \gamma_0^{(k-1)}, s^{(k)}, f^{(k)}, F^{(k)}).
$$
Theorem~\ref{thm:inverse_transformation} ensures at each step that the spectral data of $\mathscr{P}(s^{(k)}, f^{(k)}, F^{(k)})$ is $\{ \lambda_n^{(k)}, \gamma_n^{(k)} \}_{n \ge 0}$, and hence the existence part of the theorem follows.

In order to prove the uniqueness part we assume that $\mathscr{P}(s, f, F)$ and $\mathscr{P}(\widetilde{s}, \widetilde{f}, \widetilde{F})$ have the same spectral data. Then Theorem~\ref{thm:asymptotics} implies $\ind f = \ind \widetilde{f}$ and $\ind F = \ind \widetilde{F}$, and we denote $K := \max \{ \ind f, \ind F \}$. Together with $\mathscr{P}(s^{(k)}, f^{(k)}, F^{(k)})$ defined by (\ref{eq:P_k}) we consider the problems $\mathscr{P}(\widetilde{s}^{(k)}, \widetilde{f}^{(k)}, \widetilde{F}^{(k)})$ defined by
\begin{equation} \label{eq:P_tilde_k}
\begin{aligned}
  (\widetilde{s}^{(0)}, \widetilde{f}^{(0)}, \widetilde{F}^{(0)}) &:= (\widetilde{s}, \widetilde{f}, \widetilde{F}), \\
  (\widetilde{s}^{(k)}, \widetilde{f}^{(k)}, \widetilde{F}^{(k)}) &:= \widehat{\mathbf{T}} (\widetilde{s}^{(k-1)}, \widetilde{f}^{(k-1)}, \widetilde{F}^{(k-1)}), \qquad k = 1, 2, \ldots, K.
\end{aligned}
\end{equation}
Theorem~\ref{thm:transformation} yields that $\mathscr{P}(s^{(k)}, f^{(k)}, F^{(k)})$ and $\mathscr{P}(\widetilde{s}^{(k)}, \widetilde{f}^{(k)}, \widetilde{F}^{(k)})$ have the same spectral data. In particular, $\mathscr{P}(s^{(K)}, f^{(K)}, F^{(K)})$ and $\mathscr{P}(\widetilde{s}^{(K)}, \widetilde{f}^{(K)}, \widetilde{F}^{(K)})$ are two problems with constant boundary conditions and the same spectral data. Therefore, according to the discussion preceding the theorem, we have $(s^{(K)}, f^{(K)}, F^{(K)}) = (\widetilde{s}^{(K)}, \widetilde{f}^{(K)}, \widetilde{F}^{(K)})$, and successive applications of Theorem~\ref{thm:inverse} concludes the proof.
\end{proof}

\subsection{Inverse problem by two spectra} \label{ss:bytwospectra}

The results of Subsection~\ref{ss:two_problems} show that if two disjoint sequences $\{ \lambda_n \}_{n \ge 0}$ and $\{ \mu_n \}_{n \ge 0}$ are the eigenvalues of two problems of the form $\mathscr{P}(s, f, F)$ and $\mathscr{P}(s, f + \alpha, F)$, then they interlace and satisfy asymptotics of the form
\begin{equation} \label{eq:lambda_mu}
  \sqrt{\lambda_n} = n - L + \ell_2(1), \qquad \sqrt{\lambda_n} - \sqrt{\vphantom{\lambda_n} \mu_n} = (n - L)^{-2 r - 1} \left( \nu + \ell_2(1) \right)
\end{equation}
for some integer or half-integer $L \ge -1/2$, $\nu \in \mathbb{R} \setminus \{0\}$ and $r \in \{0, 1\}$, with the exception of the case when $L = -1/2$ and $r = 1$ (because of our assumption $\ind f \ge 0$, if $L = -1/2$ then necessarily $\ind f = 0$, and consequently $r = 0$). We are now going to prove that these conditions are also sufficient for two sequences to be the eigenvalues of two such problems. Note that one cannot directly apply the transformations of Section~\ref{sec:transformations} as in the previous subsection, because a pair of boundary value problems with a common boundary condition is transformed to a pair of boundary value problems with no common boundary conditions. Therefore we will first reduce our two-spectra inverse problem to the one solved in Subsection~\ref{ss:byspectraldata}.

As Theorem~\ref{thm:by_spectral_data} shows, the inverse problem by spectral data with boundary conditions of the form~(\ref{eq:boundary}), (\ref{eq:f_F}) is completely analogous to the one with constant boundary conditions in the sense that a boundary value problem is uniquely determined by its eigenvalues and norming constants. It turns out, however, that unlike the case of constant boundary conditions, a pair of problems of the form $\mathscr{P}(s, f, F)$ and $\mathscr{P}(s, f + \alpha, F)$ is not uniquely determined by their eigenvalues. One also needs to specify the poles of the function $f$. These poles (i.e., the zeros of $f_\downarrow$) are among the zeros of the difference of the corresponding characteristic functions (see (\ref{eq:xi-chi})), and our next theorem shows that they can be chosen arbitrarily among these zeros. The difference between these two kinds of inverse problems can be explained intuitively by the fact that the information about $f_\downarrow$ is already incorporated into the definition of norming constants.

\begin{theorem} \label{thm:two_spectra}
  Let $\{ \lambda_n \}_{n \ge 0}$ and $\{ \mu_n \}_{n \ge 0}$ be two interlacing sequences satisfying the asymptotics (\ref{eq:lambda_mu}). Then there exists a pair of problems of the form $\mathscr{P}(s, f, F)$ and $\mathscr{P}(s, f + \alpha, F)$ having the eigenvalues $\{ \lambda_n \}_{n \ge 0}$ and $\{ \mu_n \}_{n \ge 0}$ respectively. Moreover, there is a one-to-one correspondence between such pairs of problems and sets of nonnegative integers of cardinality not exceeding $L + (1 - r) / 2$.
\end{theorem}
\begin{proof}
Define the functions
\begin{equation*}
  \chi(\lambda) := -\prod_{n < L} (\lambda_n - \lambda) \prod_{n = L} \pi (\lambda_n - \lambda) \prod_{n > L} \frac{\lambda_n - \lambda}{(n - L)^2}
\end{equation*}
and
\begin{equation*}
  \xi(\lambda) := -\prod_{n < L} (\mu_n - \lambda) \prod_{n = L} \pi (\mu_n - \lambda) \prod_{n > L} \frac{\mu_n - \lambda}{(n - L)^2}.
\end{equation*}
Let $d$ be an integer with
\begin{equation*}
  0 \le d \le L + \frac{1 - r}{2},
\end{equation*}
and let $i_1$, $i_2$, $\ldots$, $i_d$ be integers (indices) with $0 \le i_1 < i_2 < \ldots < i_d$. Define the polynomial
\begin{equation*}
  p(\lambda) := \prod_{k=1}^d (\tau_{i_k} - \lambda),
\end{equation*}
where $\tau_0 < \tau_1 < \ldots$ are the zeros of the function $\chi(\lambda) - \xi(\lambda)$. The use of the fact that $\lambda_n$ and $\mu_n$ interlace together with (\ref{eq:xi_lambda_mu}) and (\ref{eq:chi_xi}) (which is legitimate since the derivation of these estimates used only the infinite product representations of $\chi$ and $\xi$ and the asymptotics of $\lambda_n$ and $\zeta_n$) implies that the numbers $\gamma_n$ defined by
\begin{equation*}
  \gamma_n := \frac{\pi \nu p^2(\lambda_n) \chi'(\lambda_n)}{\xi(\lambda_n)}
\end{equation*}
are all positive and have the asymptotics
\begin{equation*}
  \gamma_n = \left( n - L \right)^{4 d + 2 r} \left( \frac{\pi}{2} + \ell_2(1) \right).
\end{equation*}
By Theorem~\ref{thm:by_spectral_data}, there exists a boundary value problem $\mathscr{P}(s, f, F)$ having the eigenvalues $\{ \lambda_n \}_{n \ge 0}$ and the norming constants $\{ \gamma_n \}_{n \ge 0}$. Moreover, $\ind f = 2 d + r \ge 0$ and $\ind F = 2 L - 2 d - r \ge -1$. Denote $\alpha := \pi \nu / \left( h'_0 \right)^2$ with $h'_0$ defined as at the beginning of Section~\ref{sec:preliminaries}. It only remains to show that the problem $\mathscr{P}(s, f + \alpha, F)$ has the eigenvalues $\mu_n$. But first we show that the polynomials $f_\downarrow(\lambda)$ and $p(\lambda)$ coincide up to a constant factor. Arguing as in the proof of Lemma~\ref{lem:zero_sum} we have
\begin{equation*}
  \sum_{n=0}^\infty \frac{\lambda_n^k p(\lambda_n)}{\gamma_n} = \sum_{n=0}^\infty \frac{\lambda_n^k \xi(\lambda_n)}{\pi \nu p(\lambda_n) \chi'(\lambda_n)} = \frac{1}{2 \pi^2 \nu \iu} \lim_{N \to \infty} \int_{C_N} \frac{\lambda^k \left( \xi(\lambda) - \chi(\lambda) \right)}{p(\lambda) \chi(\lambda)} \,\du \lambda = 0,
\end{equation*}
where $C_N$ is the same as in that proof. Thus, by the same lemma, $f_\downarrow(\lambda) = h'_0 p(\lambda)$.

Denote the eigenvalues of the boundary value problem $\mathscr{P}(s, f + \alpha, F)$ by $\widehat{\mu}_n$. They coincide with the zeros of the function
\begin{equation*}
  \widehat{\xi}(\lambda) := F_\uparrow(\lambda) \theta(\pi, \lambda) - F_\downarrow(\lambda) \theta^{[1]}_s(\pi, \lambda),
\end{equation*}
where $\theta(x, \lambda)$ is defined as in~(\ref{eq:theta}). Using the results of Subsection~\ref{ss:two_problems}, we obtain
\begin{equation*}
  \widehat{\xi}(\lambda_n) = \frac{\alpha f_\downarrow^2(\lambda_n) \chi'(\lambda_n)}{\gamma_n} = \frac{\pi \nu p^2(\lambda_n) \chi'(\lambda_n)}{\gamma_n} = \xi(\lambda_n), \qquad n \ge 0.
\end{equation*}
This and the asymptotics of $\chi$, $\xi$ and $\widehat{\xi}$ show that $\left( \widehat{\xi}(\lambda) - \xi(\lambda) \right) / \chi(\lambda)$ is an entire function satisfying the estimate
\begin{equation*}
  \frac{\widehat{\xi}(\lambda) - \xi(\lambda)}{\chi(\lambda)} = o(1)
\end{equation*}
on $\bigcup_N C_N$ and hence by the maximum principle on the whole plane. Then the Liouville theorem yields that this function is identically zero. Thus $\widehat{\xi}(\lambda) \equiv \xi(\lambda)$ and hence $\widehat{\mu}_n = \mu_n$, $n \ge 0$.

So far, we have constructed two problems $\mathscr{P}(s, f, F)$ and $\mathscr{P}(s, f + \alpha, F)$ with the eigenvalues $\{ \lambda_n \}_{n \ge 0}$ and $\{ \mu_n \}_{n \ge 0}$ respectively, and such that the poles of $f$ are $i_1$'s, $i_2$'s, $\ldots$, $i_d$'s zeros of $\chi(\lambda) - \xi(\lambda)$. To prove that such a pair of problems is unique, we assume that the problems $\mathscr{P}(s, f, F)$ and $\mathscr{P}(\widetilde{s}, \widetilde{f}, \widetilde{F})$ have the eigenvalues $\{ \lambda_n \}_{n \ge 0}$, and the problems $\mathscr{P}(s, f + \alpha, F)$ and $\mathscr{P}(\widetilde{s}, \widetilde{f} + \widetilde{\alpha}, \widetilde{F})$ have the eigenvalues $\{ \mu_n \}_{n \ge 0}$. We also assume that the poles of $f$ (respectively, $\widetilde{f}$) are $i_1$'s, $i_2$'s, $\ldots$, $i_d$'s zeros of $\chi - \xi$ (respectively, $\widetilde{\chi} - \widetilde{\xi}$). Then $\chi \equiv \widetilde{\chi}$ and $\xi \equiv \widetilde{\xi}$ by the definition of these functions, and hence $\left( h'_0 \right)^{-1} f_\downarrow \equiv \left( \widetilde{h}'_0 \right)^{-1} \widetilde{f}_\downarrow$. On the other hand, the asymptotics of the eigenvalues yields $\alpha \left( h'_0 \right)^2 = \widetilde{\alpha} \left( \widetilde{h}'_0 \right)^2$ (see the formula for $\sqrt{\lambda_n} - \sqrt{\vphantom{\lambda_n} \mu_n}$ at the end of Subsection~\ref{ss:two_problems}). Thus
\begin{equation*}
  \gamma_n = \frac{\alpha f_\downarrow^2(\lambda_n) \chi'(\lambda_n)}{\xi(\lambda_n)} = \frac{\widetilde{\alpha} \widetilde{f}_\downarrow^2(\lambda_n) \widetilde{\chi}'(\lambda_n)}{\widetilde{\xi}(\lambda_n)} = \widetilde{\gamma}_n.
\end{equation*}
Therefore the uniqueness part of Theorem~\ref{thm:by_spectral_data} implies that $s = \widetilde{s}$ a.e. on $[0, \pi]$, $f = \widetilde{f}$ and $F = \widetilde{F}$. Finally, $f = \widetilde{f}$ yields $\alpha = \pi \nu / \left( h'_0 \right)^2 = \pi \nu / \left( \widetilde{h}'_0 \right)^2 = \widetilde{\alpha}$.
\end{proof}

\begin{remark}
In particular, this proof shows that the problems $\mathscr{P}(s, f, F)$ and $\mathscr{P}(s, f + \alpha, F)$ are uniquely determined by their spectra and the poles of $f$. Theorem~\ref{thm:two_spectra} also yields that the two spectra determine these problems uniquely if and only if $\ind F \le 0 \le \ind f \le 1$ (i.e., the second boundary condition does not contain the eigenvalue parameter at all and the first boundary condition may depend on it only linearly).
\end{remark}

\subsection{Inverse problems by one spectrum} \label{ss:byonespectrum}

Theorems~\ref{thm:by_spectral_data} and \ref{thm:two_spectra} show that the spectrum of a boundary value problem of the form (\ref{eq:SL})-(\ref{eq:boundary}) does not uniquely determine this problem. In order for the unique determination by one spectrum to work, one has to impose some additional restrictions. In this subsection we will consider two types of such restrictions. We will call a boundary value problem of the form $\mathscr{P}(s, f, f)$ \emph{symmetric} if $s(x) + s(\pi - x) = 0$ (which, for differentiable $s \in \mathring{\mathscr{L}}_2(0, \pi)$, is equivalent to $s'(x) = s'(\pi - x)$). In the first part of the subsection we will prove that the spectrum alone determines the symmetric problem $\mathscr{P}(s, f, f)$.

We start by studying the properties of symmetric problems. Theorem~\ref{thm:asymptotics} shows that the eigenvalues of $\mathscr{P}(s, f, f)$ satisfy the asymptotics
\begin{equation} \label{eq:symmetric_asymptotics}
  \sqrt{\lambda_n} = n - L + \ell_2(1),
\end{equation}
where $L := \ind f$. Since our problem is symmetric, it follows from~(\ref{eq:phi_psi}) that $\psi(x, \lambda) = \varphi(\pi - x, \lambda)$. Then~(\ref{eq:beta}) implies $\psi(x, \lambda_n) = \beta_n^2 \psi(x, \lambda_n)$, and hence $\beta_n^2 = 1$. Using Theorem~\ref{thm:oscillation} we obtain $\beta_n = (-1)^n$. Thus (\ref{eq:chi_beta_gamma}) implies
\begin{equation} \label{eq:symmetric_chi_gamma}
  \gamma_n = (-1)^n \chi'(\lambda_n).
\end{equation}
Since $L$ is now an integer, the formula~(\ref{eq:infinite_product}) takes the form
\begin{equation} \label{eq:product}
  \chi(\lambda) = -\pi \prod_{n=0}^L (\lambda_n - \lambda) \prod_{n=L+1}^\infty \frac{\lambda_n - \lambda}{(n - L)^2}.
\end{equation}

Now we are ready to state the first result of this subsection.

\begin{theorem} \label{thm:symmetric}
Let $\{ \lambda_n \}_{n \ge 0}$ be a strictly increasing sequence of real numbers satisfying the asymptotics (\ref{eq:symmetric_asymptotics}) for some integer $L \ge -1$. Then there exists a unique symmetric boundary value problem $\mathscr{P}(s, f, f)$ having the spectrum $\{ \lambda_n \}_{n \ge 0}$.
\end{theorem}
\begin{proof}
Define $\chi$ by~(\ref{eq:product}) and then $\gamma_n$ by~(\ref{eq:symmetric_chi_gamma}). The expression (\ref{eq:product}) and the use of (\ref{eq:chi_xi}) in (\ref{eq:symmetric_chi_gamma}) imply that $\gamma_n$ are strictly positive numbers satisfying the asymptotics
\begin{equation*}
  \gamma_n = \frac{\pi}{2} (n - L)^{2L} \left( 1 + \ell_2(1) \right).
\end{equation*}
The rest of the proof now follows from Theorem~\ref{thm:by_spectral_data}.
\end{proof}

Another type of inverse problems where one spectrum is sufficient for the unique determination is a class of problems known under the names of \emph{problems with mixed given data}, \emph{problems with partial information on the potential} or \emph{half-inverse problems}. For problems with summable potentials and constant boundary conditions, the Hochstadt--Lieberman theorem states that the knowledge of the potential on $[0, \pi / 2]$ together with the boundary coefficient at $0$ and the spectrum uniquely determines the other boundary coefficient and the potential a.e. on $[\pi / 2, \pi]$. Hryniv and Mykytyuk \cite[Theorem 2.1]{HM04b} showed that this result holds also in the case of distributional potentials with constant boundary conditions. In our notation, (the uniqueness part of) \cite[Theorem 2.1]{HM04b} states that if the spectra of the problems $\mathscr{P}(s, h, H)$ and $\mathscr{P}(\widetilde{s}, \widetilde{h}, \widetilde{H})$ with constant $h$, $H$, $\widetilde{h}$, $\widetilde{H}$ coincide and $s(x) - h = \widetilde{s}(x) - \widetilde{h}$ a.e. on $[0, \pi / 2]$, then $s(x) - h = \widetilde{s}(x) - \widetilde{h}$ a.e. on $[\pi / 2, \pi]$ and $H + h = \widetilde{H} + \widetilde{h}$. In this subsection we will generalize this result to the case of boundary value problems of the form (\ref{eq:SL})-(\ref{eq:boundary}).

We start with an auxiliary lemma.

\begin{lemma} \label{lem:half_inverse}
Suppose that $\mathscr{P}(s_1, f_1, F_1)$ and $\mathscr{P}(s_2, f_2, F_2)$ with $f_1 \ne \infty \ne f_2$ have the same spectra, $(\widehat{s}_1, \widehat{f}_1, \widehat{F}_1) = \widehat{\mathbf{T}} (s_1, f_1, F_1)$, $(\widehat{s}_2, \widehat{f}_2, \widehat{F}_2) = \widehat{\mathbf{T}} (s_2, f_2, F_2)$, and $a \in (0, \pi]$ is an arbitrary real number. Then $f_1(\lambda) - f_2(\lambda) = \const = s_1(x) - s_2(x)$ for a.e. $x \in [0, a]$ if and only if $\widehat{f}_1(\lambda) - \widehat{f}_2(\lambda) = \const = \widehat{s}_1(x) - \widehat{s}_2(x)$ for a.e. $x \in [0, a]$.
\end{lemma}
\begin{proof}
To prove the necessity we notice that $f_1(\lambda) - f_2(\lambda) = \const$ yields $\ind f_1 = \ind f_2$, and hence the asymptotics of the eigenvalues (see Theorem~\ref{thm:asymptotics}) implies $\ind F_1 = \ind F_2$. Denote
\begin{equation*}
  \Lambda := \begin{cases} \mathring{\boldsymbol{\uplambda}}(q_1, f_1, F_1), & F_1 \ne \infty, \\ \mathring{\boldsymbol{\uplambda}}(q_1, f_1, F_1) - 2, & F_1 = \infty \end{cases} = \begin{cases} \mathring{\boldsymbol{\uplambda}}(q_2, f_2, F_2), & F_2 \ne \infty, \\ \mathring{\boldsymbol{\uplambda}}(q_2, f_2, F_2) - 2, & F_2 = \infty. \end{cases}
\end{equation*}
Then the solutions $v_1(x)$ and $v_2(x)$ of the initial value problems
\begin{equation*}
  - \left( y^{[1]}_{s_1} \right)' - s_1 y^{[1]}_{s_1} - s_1^2 y = \Lambda y, \qquad v_1(0) = \left( f_1 \right)_\downarrow(\Lambda), \qquad \left( v_1 \right)^{[1]}_{s_1}(0) = - \left( f_1 \right)_\uparrow(\Lambda)
\end{equation*}
and
\begin{equation*}
  - \left( y^{[1]}_{s_2} \right)' - s_2 y^{[1]}_{s_2} - s_2^2 y = \Lambda y, \qquad v_2(0) = \left( f_2 \right)_\downarrow(\Lambda), \qquad \left( v_2 \right)^{[1]}_{s_2}(0) = - \left( f_2 \right)_\uparrow(\Lambda)
\end{equation*}
coincide on $[0, a]$. Therefore (\ref{eq:s_f_F_hat}) implies
\begin{equation*}
  \widehat{f}_1(\lambda) - \widehat{f}_2(\lambda) = \frac{2}{\pi} \ln \frac{v_1(\pi)}{v_2(\pi)} = \widehat{s}_1(x) - \widehat{s}_2(x) \quad \text{for a.e. $x \in [0, a]$.}
\end{equation*}

The sufficiency can be proved similarly.
\end{proof}

Applying this lemma successively to the problems defined by (\ref{eq:P_k}) and (\ref{eq:P_tilde_k}) with $a = \pi / 2$, using the above-mentioned result of Hryniv and Mykytyuk, and applying the lemma again to (\ref{eq:P_k}) and (\ref{eq:P_tilde_k}) in reverse order with $a = \pi$, we obtain the following generalization of the Hochstadt--Lieberman theorem.

\begin{theorem} \label{thm:partial}
If the spectra of the problems $\mathscr{P}(s, f, F)$ and $\mathscr{P}(\widetilde{s}, f, \widetilde{F})$ with $\ind f \ge \ind F$ coincide and $s(x) = \widetilde{s}(x)$ a.e. on $[0, \pi / 2]$, then $s(x) = \widetilde{s}(x)$ a.e. on $[\pi / 2, \pi]$ and $F(\lambda) = \widetilde{F}(\lambda)$.
\end{theorem}

\end{document}